\documentclass[12pt]{amsart}
\usepackage{amsfonts,amsmath,amssymb}
\usepackage{wasysym}
\usepackage{verbatim}
\usepackage[dvips]{graphicx}
\usepackage{psfrag}
\usepackage{epsfig}
\usepackage{color}
\usepackage{graphics}
\usepackage{subfig}
\usepackage{pdfsync}
\definecolor{ddorange}{rgb}{1,0.5,0}

\theoremstyle{plain}
\newtheorem{theorem}{Theorem}[section]
\newtheorem{lemma}[theorem]{Lemma}

\newtheorem{corollary}[theorem]{Corollary}
\newtheorem{remark}[theorem]{Remark}
\newtheorem{definition}[theorem]{Definition}

\theoremstyle{definition}
\theoremstyle{remark}
\numberwithin{equation}{section}
\mathchardef\emptyset="001F

\newcommand{\e}{\varepsilon}
\newcommand{\Om}{\Omega}

\newcommand{\R}{{\mathbb R}}
\newcommand{\F}{{\mathcal F}}
\newcommand{\G}{{\mathcal G}}

\newcommand{\E}{{\mathcal E}}

\newcommand{\A}{{\mathcal A}} 
\renewcommand{\E}{{\mathcal E}}
\renewcommand{\P}{{\mathcal P}^9}

\newcommand{\Z}{{\mathbb Z}}
\newcommand{\N}{{\mathbb N}}

\renewcommand{\H}{{\mathcal H}}

\newcommand{\dist}{\mathrm{dist}}

\newcommand{\be}{\begin{equation}}
\newcommand{\ee}{\end{equation}}
\newcommand{\bes}{\begin{equation*}}
\newcommand{\ees}{\end{equation*}}
\newcommand{\bea}{\begin{eqnarray}}
\newcommand{\eea}{\end{eqnarray}}
\newcommand{\beas}{\begin{eqnarray*}}
\newcommand{\eeas}{\end{eqnarray*}}

\renewcommand{\mod}{\mathop\mathrm{mod}}

\begin{document}
\title[Interfacial energies of systems of chiral molecules
]{Interfacial energies of systems of chiral molecules
}
\author{Andrea Braides}
\author{Adriana Garroni}
\author{Mariapia Palombaro}
\address
{Andrea Braides: Dipartimento di Matematica, Universit\`a di Roma Tor Vergata, via della ricerca scientifica 1, 00133 Roma, Italy
}
\email{braides@mat.uniroma2.it}
\address
{Adriana Garroni: Dipartimento di Matematica, Sapienza Universit\`a di Roma, 00185 Roma, Italy
}
\email{garroni@mat.uniroma1.it}
\address{
Mariapia Palombaro: University of Sussex, Department of Mathematics, Pevensey 2 Building, Falmer Campus,
Brighton BN1 9QH, United Kingdom}
\email{M.Palombaro@sussex.ac.uk
}
%
%
\begin{abstract}
We consider a simple model for the assembly of chiral molecules in two dimensions driven by maximization of the contact area. We derive a macroscopic model described by a parameter taking nine possible values corresponding to the possible minimal microscopic patterns and modulated phases of the chiral molecules. We describe the overall behaviour by means of an interaction energy of perimeter type between such phases. This energy is a crystalline perimeter energy, highlighting preferred directions for the interfaces between ensembles of molecules labelled by different values of the parameter.
\end{abstract}
\maketitle
{\small
\par

\bigskip
\section*{Introduction}
We consider a simple model of interaction between ensembles of two types of chiral molecules in two dimensions. The model has been described, together with other ones, in a paper by P.~Szabelski and A.~Woszczyk \cite{SW} (see also \cite{BSG}). 
\begin{figure}[h!]
\centering
\includegraphics[width=.33\textwidth]{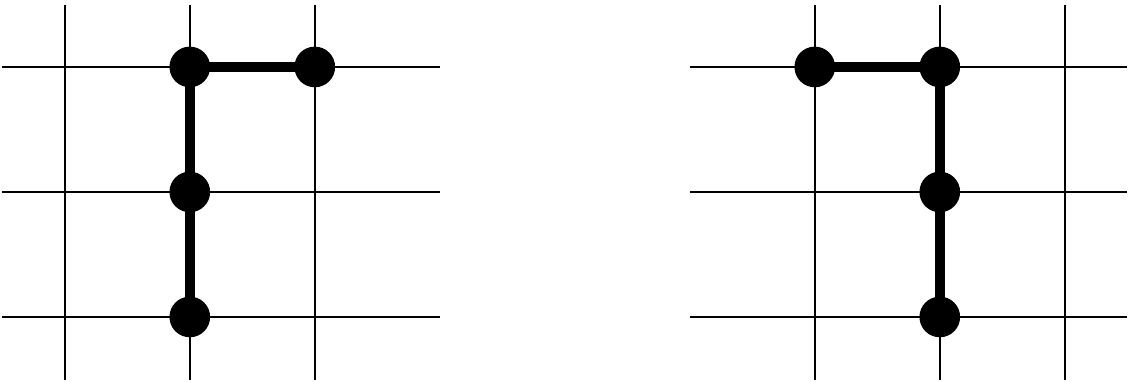}
\caption{Schematic picture of chiral molecules.}\label{chi-mol}
\end{figure}
Such molecules are considered as occupying four sites of a square lattice with three points aligned in the vertical direction and one on one side in the two fashions represented in Fig.~\ref{chi-mol}.
We consider collections of molecules; i.e., collections of disjoint sets of points, each of which
differs from one of the two just described by an integer translation.
In \cite{SW} the energy of a molecule is expressed as
\begin{equation}\label{enSW}
\sum_{i=1}^4\sum_{j=1}^4s_{ij}\ ,
\end{equation}
where $i$ parameterizes the four sites composing a molecule, and $j$ parameterizes
the four neighbouring sites in $\Z^2$ of each element of the molecule; the value $s_{ij}$ equals $0$ if the site parameterized by $j$ belongs to some molecule (in particular it is $0$ if that site belongs to the same molecule)
and equals $1$ if it does not belong to any other molecule. A discussion about chemical mechanisms for such energies can be found in \cite{P,PSR}. 
Our model can be viewed as a lattice system. Indeed, (up to an additive constant) the energy density $s_{ij}$ in \eqref{enSW} is nothing else than a ferromagnetic spin energy if we define the spin variable equal to $1$ on the sites of the molecules and equal to $-1$ on the remaining sites of a square lattice.

The objective of our analysis is to give a homogenized description of such a system through an approximate macroscopic energy which describes the typical collective mechanical behaviour of chiral molecules (see, e.g., \cite{FPE,H,Co}). 
The usual representation of the overall properties by a macroscopic spin variable or magnetization
\cite{P} is not possible, since the geometry of the system makes it of non-local type, as the fact that a site is occupied by a molecule of either type influences the system further than the nearest neighbours in a non-trivial way.
Moreover, such a simple representation would integrate out the asymmetric microscopic behaviours of the molecules.
For a better description, we have to define a new parameter
that captures the relevant properties of the microscopic arrangement of the molecules, 
in the spirit of recent works on lattice systems with microstructure \cite{ABC,BC-last}.

In order to define an overall macroscopic parameter 
we first note that we may equivalently represent chiral molecules as unions of unit squares centered on points of 
a square lattice as in \begin{figure}[ht]
\centering
\includegraphics[width=.27\textwidth]{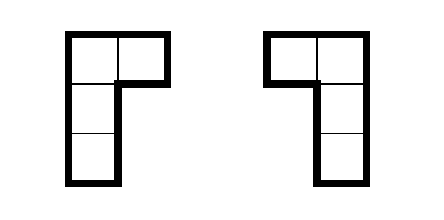}
\caption{Chiral molecules as union of squares.}\label{fig:molecole-chirali}
\end{figure}
Fig.~\ref{fig:molecole-chirali}. 
Correspondingly, the energy in \eqref{enSW} can
be viewed as the length of the boundary of the molecule not in contact with any other molecule.
The energy of a collection of molecules is then simply the total length of the boundary of the union of the molecules.
\begin{figure}[ht]
\centering
\includegraphics[width=.5\textwidth]{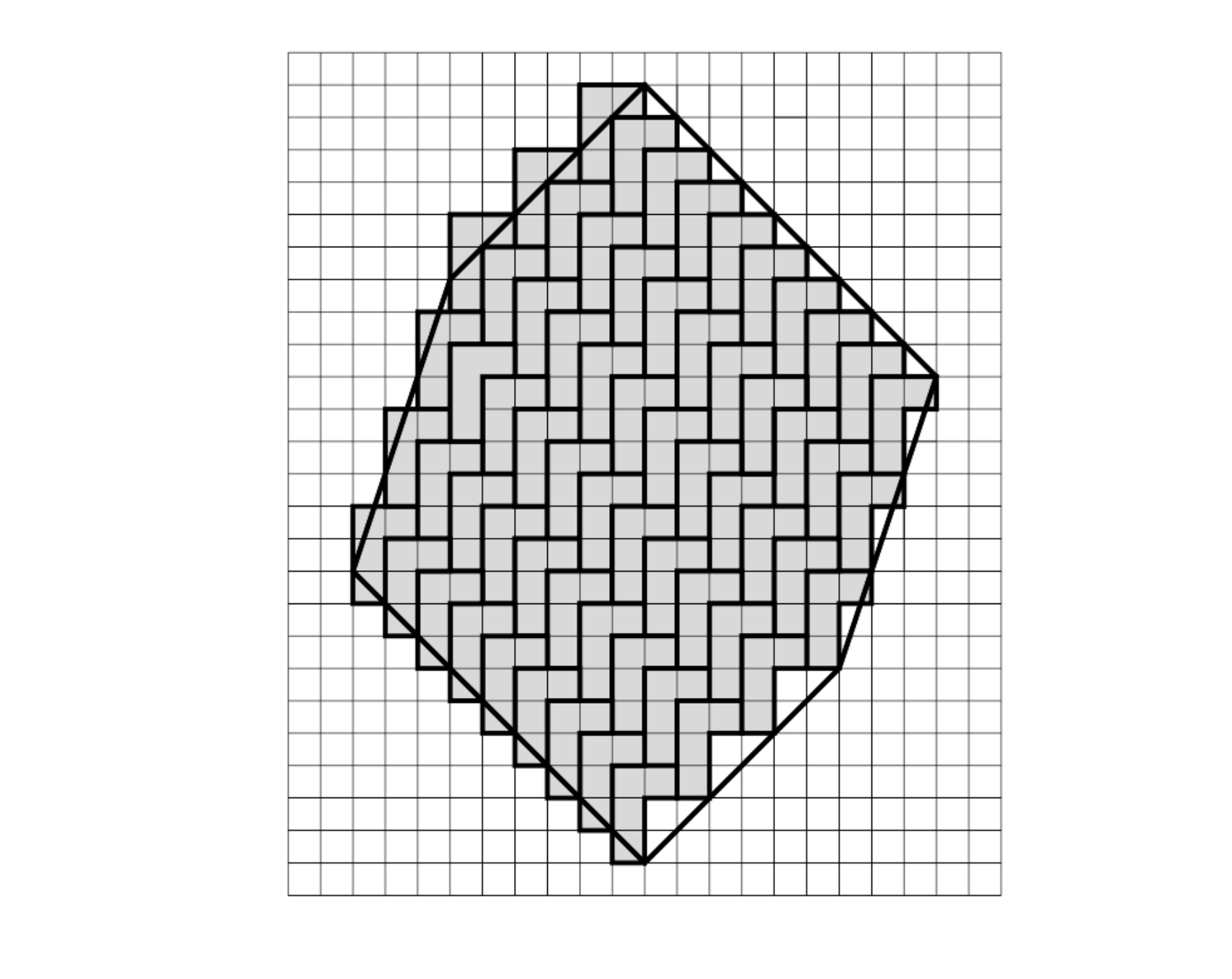}
\caption{An ensemble of a single type of molecules minimizing the perimeter of their union, and its continuous approximation}
\label{Fig25}
\end{figure}
We then examine the patterns of 
sets with minimal energy. In Fig.~\ref{Fig25} we picture a set composed of  a given number of one type of molecules
minimizing its total boundary length. 
We first make the simple observation that, whenever this is allowed by boundary conditions, configuration
of minimal energy replicate the pattern exhibited in that figure.
\begin{figure}[ht]
\centering
\includegraphics[width=.50\textwidth]{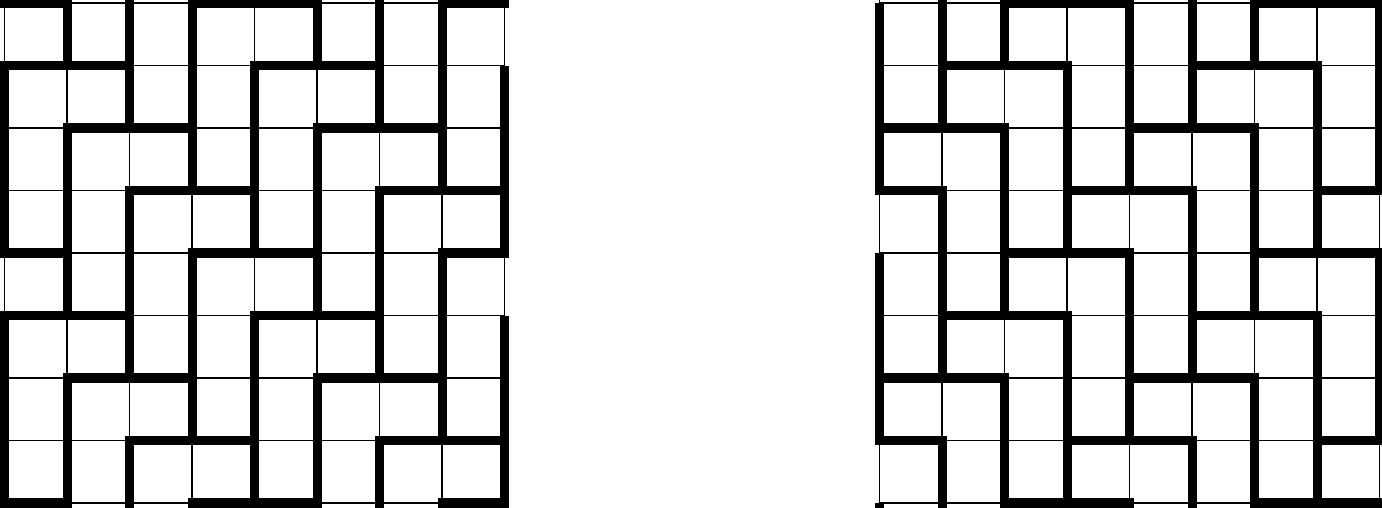}
\caption{Patterns of configurations with zero energy in a square}\label{Fig2}
\end{figure}
More precisely, we prove that configurations with zero energy inside an open set
either are the empty set (no molecule is present) or in the interior of that set they must correspond to a ``striped'' pattern of either of the two types in Fig.~\ref{Fig2}. These two patterns only are not sufficient to describe the behaviour of our energy, since the simultaneous presence of different translations of the same pattern will result in interstitial voids, and hence will have non-zero energy. We then remark that each pattern is four-periodic, and
translating the pattern vertically (or horizontally) we obtain all different arrangements with zero energy. As a consequence, for each pattern we have four ``modulated phases'' corresponding to the four translations.
\begin{figure}[ht]
\centering
\includegraphics[width=.12\textwidth]{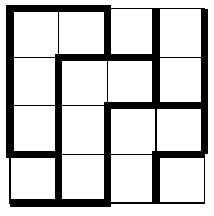}\qquad
\includegraphics[width=.12\textwidth]{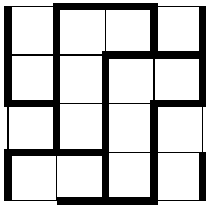}\qquad
\includegraphics[width=.12\textwidth]{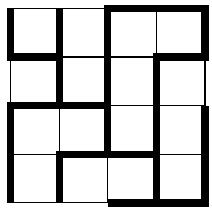}\qquad
\includegraphics[width=.12\textwidth]{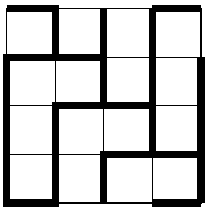}\qquad
\caption{Four distinct modulated phases 
restricted to the same periodicity cell}
\label{Fig3}
\end{figure}
In Fig.~\ref{Fig3} we reproduce the unit cells of the different phases corresponding to the same pattern.
In this way we have singled out eight different arrangements for the ground state, to which we have to add the trivial configuration with zero energy corresponding to the empty set. Note that the position of a single molecule determines the corresponding ground state.

In order to study the overall behaviour of a system of such chiral molecules, we 
follow a discrete-to-continuum approach, by scaling the system by a small parameter $\e$ and examine its behaviour as $\e\to0$.
We first give a notion of convergence of a family $E_\e$ of sets which are unions of scaled molecules of disjoint interior to a family $A_1,\ldots,A_8$ by decomposing the set $E_\e$ into the sets $E_{1,\e},\ldots,E_{8,\e}$ defined as the union of the molecules corresponding to one of the eight modulated phases, respectively, and requiring that the symmetric difference between $E_{j,\e}$ and the corresponding $A_j$ tends to $0$ on each compact set of ${\mathbb R}^2$. 
We prove that this notion is indeed compact: if we have a family $E_\e$ of such sets with boundary with equibounded length, then, up to subsequences, it converges in the sense specified above. This is a non-trivial fact since it derives from a bound on the length of the boundary of the union of the sets $E_{j,\e}$, and not on each subfamily separately.
We can nevertheless prove that each family $E_{j,\e}$ satisfies a similar bound on the length of the boundaries, and as a consequence is pre-compact as a family of sets of finite perimeter. 

We then turn our attention to the description of the limit behaviour of the energies ${\mathcal F}_\e(E)={\mathcal H}^1(\partial E)$ defined on unions of scaled molecules
with respect to the convergence $E_\e\to(A_1,\ldots, A_8)$ defined above.
It is convenient to introduce the set $A_0$, complement of the union of $A_1,\ldots, A_8$, which then corresponds to the limit of the complements of $E_\e$. In this way the completed family $\{A_0,\ldots, A_8\}$ is a partition into sets of finite perimeter, for whose interfacial energies there exists an established variational theory \cite{AB}. We then represent the
$\Gamma$-limit of the energies ${\mathcal F}_\e$ as 
$$
\F(A_0,\ldots, A_8)= \sum_{i\neq j=0}^8\int_{\partial A_i\cap\partial A_j} f_{ij} (\nu^i)d\H^{1},
$$
where $f_{ij}$ is an interfacial energy and $\nu^i$ is the measure-theoretical normal to 
$\partial A_i$. The functions $f_{ij}(\nu)$ are represented by an asymptotic homogenization formula which describes the optimal way to microscopically arrange the molecules between two macroscopic phases $A_i$ and $A_j$ in a way to obtain an average interface with normal $\nu$. Note that this optimization process may be achieved by the use of molecules corresponding to phases other than $A_i$ and $A_j$.

This process can be localized, requiring that all molecules be contained in a set $\Omega$. In this case the same description holds, upon requiring that the partition satisfies $\bigcup_{j=1}^8 A_j\subset \Omega$, or equivalently $A_0\supset {\mathbb R}^2\setminus \Omega$.

With the aid of the homogenization formula, we are able to actually compute the energy densities
$f_i=f_{i0}=f_{0i}$; i.e., with one of the two phases corresponding to the empty set. In that case, $f_i$ is a crystalline perimeter energy, whose Wulff shape is an irregular hexagon corresponding
to the continuous approximation of sets as in Fig.~\ref{Fig25}.

The paper is organized as follows. In Section \ref{Sect1} 
we introduce the necessary notation and prove the geometric Lemma \ref{quadrato-interno} 
which characterizes configurations with zero energy on an open set. With the aid of 
that result in Section \ref{Sect2} we define the discrete-to continuous convergence of 
scaled families of chiral molecules to partitions into nine sets of finite perimeter, and prove that
this is a compact convergence on families with equibounded energy. In Section \ref{Sect3} we
first define the limit interfacial energy densities through an asymptotic homogenization formula
and subsequently prove the $\Gamma$-convergence of the energies on scaled chiral systems
to the energy defined through those interfacial energy densities. We then compute the energy densities and the related Wulff shapes when one of the phases is the empty set, and describe the 
treatment of  anchoring boundary conditions. Finally, Section \ref{Sect4} is dedicated 
to generalization; in particular we remark that we may include a dependence
on the type of chiral molecule, in which case optimal configuration may develop wetting layers.
Another interesting observation is that we may consider as model energy the two-dimensional measure of $\R^2$ not occupied by a system of molecules (scaled by $1/\e$ for dimensional
reasons when scaling the molecules) in place of the one-dimensional measure of their boundary.
The analysis proceeds with minor changes except for the fact that the domain of the limit is
restricted to the eight non-empty phases.

\section{Geometric setting}\label{Sect1}
We will consider  $\R^2$ equipped with the usual scalar product, for which we use 
the notation $x\cdot y$. The Lebesgue measure of a set $E$ will be denoted by $|E|$; its 
$1$-dimensional Hausdorff measure by $\H^1(E)$.
Given $U\subset\R^2$ and  $x\in\R^2$, we denote by $U(x)$ the translation of $U$ 
by $x$; namely, $U(x):= x + U$.

\smallskip
We introduce the two {\em fundamental chiral molecules} as
$$
R:=\big( [0,1]\times[0,3]\big) \cup \big([1,2]\times [2,3]\big)\,, \qquad 
S:=\big(  [-1,0]\times[0,3]\big)  \cup \big(  [-2,-1]\times [2,3] \big) \,,
$$
corresponding to the two shapes in Fig.~\ref{fig:molecole-chirali}.
We will consider sets that can be obtained as a union of integer translations 
of one of these two cells with pairwise disjoint interior. We denote by ${\mathcal E}$ the collection of  families of sets defined as
$$
\mathcal E:= 
\left\{ 
\{E_j \}_j: 
E_{j} \in \{R(n),  S(n): \,n\in\Z^2\}, 
|E_{j}\cap E_{j'}|= 0  \hbox{ if } j\neq j'
\right\}.
$$
In this notation we do not specify the set of the indices $j$ since it will never be relevant in our arguments.
We may also simply write $\{E_j \}$ in the place of $\{E_j \}_j$ if no ambiguity arises. Each set of $E_j$ will
be referred to as a {\em molecule}.

We define $\A$ as the family of sets defined as unions of families in $\mathcal E$
$$
\A:= 
\Biggl\{ 
\bigcup_j E_j : 
\{E_j\}_j\in {\mathcal E}
\Biggr\}.
$$

We will sometimes need to define the union of the elements of a family $\mathcal B$ of sets.
In this case we simply write $\bigcup \mathcal B$ for $\bigcup_{B\in\mathcal B}B$. In particular,
then, $\bigcup\{E_j\}=\bigcup_j E_j$.

\smallskip

\def\ZZ{{\mathcal Z}}
In order to define the relevant macroscopic order parameter of the system,
we now prove that if a set $E\in{\mathcal A}$ has no boundary inside a (sufficiently
large) set, then, it must coincide with one single variant of a ground state as 
defined in the Introduction. In order to better formalize this property,
for each $i\in\{1,\dots, 8\}$, we introduce the family $\ZZ_i$
defined by  
\bea
\ZZ_i:= 
\begin{cases}
\{  R(n) : n\in\Z^2, n_2 + n_1 \equiv i \, \mod 4 \}   \text{ for }   i=1,\dots, 4,\\
 \{  S(n): n\in\Z^2, n_2 - n_1 \equiv i \, \mod 4 \}    \text{ for } i=5,\dots, 8.
\end{cases}
\eea

Clearly, it suffices to prove this property for squares.

For each $x\in\R^2$, $Q_{r}(x)$ stands for the open square of center $x$ and side length $r$.
In the case when $x=(0,0)$, we will simply write $Q_r$.

%
\begin{lemma}\label{quadrato-interno}
Let $\{E_j\}_j\in\E$, $n\in\Z^2, k\in\N$ with $k\geq 4$, and let $E = \bigcup E_j$. Suppose that $E \cap Q_{2k}(n)= Q_{2k}(n)$.
Then there exists  $i\in\{1,\dots, 8\}$ such that $E_j\in \ZZ_i$
for each $j$ such that $E_j\cap  Q_{2k{-4}}(n) \neq \emptyset$. 
\end{lemma}
\begin{proof}

{\it Step $1$}. Let $E_{\bar{\j}} \subset Q_{2k}(n)$. 
\begin{figure}[ht]
\centering
\includegraphics[width=.6\textwidth]{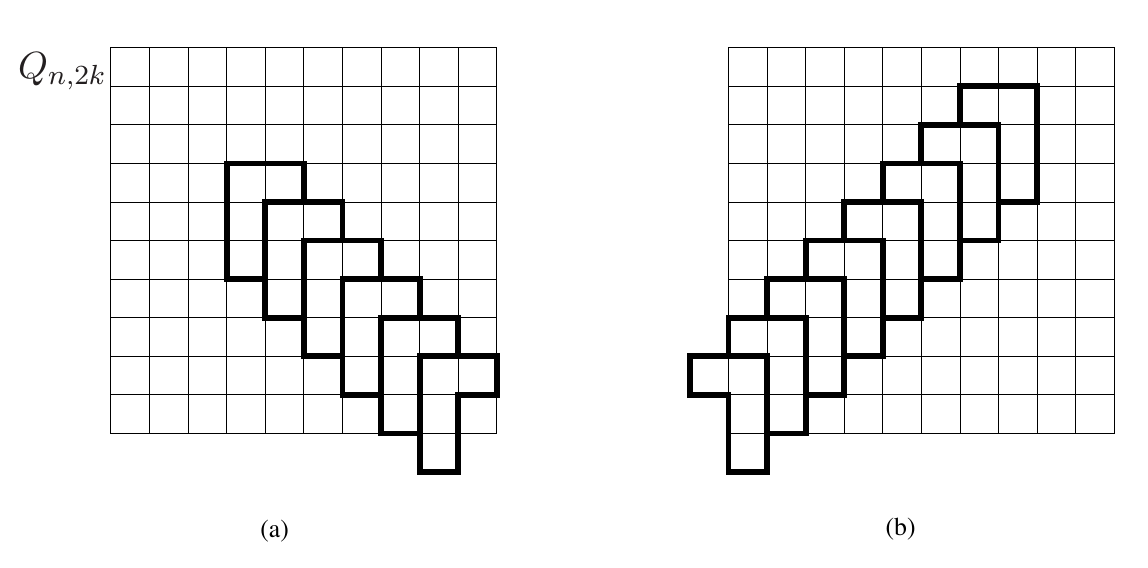}
\caption{Diagonal minimal patterns}
\label{Fig4}
\end{figure}
If $E_{\bar{\j}}$ is a translation of $R$, say $R(\bar{\j}_1,\bar{\j}_2)$,
then $R(\bar{\j}_1+1,\bar{\j}_2-1)$ (the translation of $E_{\bar{\j}}$
by $(1,-1)$) is also part of the family $\{E_j\}$. Indeed if it were not so
then the square $[0,2]\times[0,2]+(\bar{\j}_1,\bar{\j}_2)$ 
(i.e., the square defining the upper-left corner of $R(\bar{\j}_1+1,\bar{\j}_2-1)$)
would belong to an element of the form $S(\bar{\i}_1,\bar{\i}_2)$.
But then the square $[0,2]\times[0,1]+(\bar{\j}_1,\bar{\j}_2)$ would not belong 
to $E$, which is a contradiction. By proceeding by induction we deduce
that among the sets $E_j$ there are all the translations of $E_{\bar{\j}}$
in direction $(1,-1)$ contained in $Q_{2k+2}(n)$; i.e.,
all the sets of the form $R(\bar{\j}_1+t,\bar{\j}_2-t)$, with $t\in\N$, as long as 
$R(\bar{\j}_1+t,\bar{\j}_2-t)\subset Q_{2k+2}(n)$ (see Fig.~\ref{Fig4}a).

{\em Step $2$}. We now prove that also the translations $R(\bar{\j}_1-t,\bar{\j}_2+t)$ with $t\in\N$
(i.e., also the translations of $E_{\bar{\j}}$ in direction $(-1,1)$) belong to the family
$\{E_j\}$ as long as $R(\bar{\j}_1-t,\bar{\j}_2+t)\subset Q_{2k-2}(n)$. 
We can proceed by finite induction. It suffices to consider the case $R(\bar{\j}_1,\bar{\j}_2)\subset Q_{2k-4}(n)$ 
and prove that $R(\bar{\j}_1-1,\bar{\j}_2+1)$ belongs to 
the family $\{E_j\}$. 
\begin{figure}[ht]
\centering
\includegraphics[width=.9\textwidth]{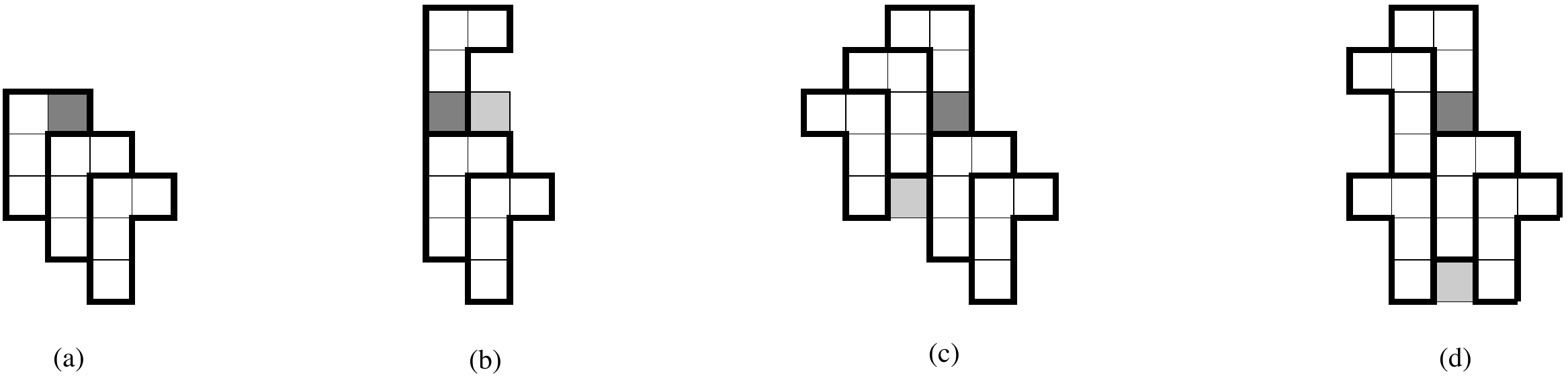}
\caption{Necessity of upper-left translations}
\label{Fig5}
\end{figure}
We suppose otherwise and argue by contradiction.
Referring to Fig.~\ref{Fig5} for a visual interpretation of
the proof, we note that the square $Q=[\bar{\j}_1, \bar{\j}_1+1]\times [\bar{\j}_2+3, \bar{\j}_2+4]$ (the dark gray square in Fig.~\ref{Fig5}) belongs to $E$. 
If it belonged to a molecule $R(\bar{\i})$ (the case in Fig.~\ref{Fig5}(b))
then this molecule should be $R(\bar{\j}_1,\bar{\j}_2+3)$. In this case, the neighbouring square $[\bar{\j}_1+1, \bar{\j}_1+2]\times [\bar{\j}_2+3, \bar{\j}_2+4]$ would not belong to $E$. Since this is not the case, $Q$ must belong to some molecule $S(\bar{\i})$ belonging to $\{E_j\}$, which is the one pictured in Fig.~\ref{Fig5} {(}c) and (d). 
Then also $S(\bar{\i}+(-1,-1))$ must belong to $\{E_j\}$ (for the same reasoning as
in Step $1$). We then have two possibilities, pictured in Fig.~\ref{Fig5}{(}c) and (d), respectively: either $S(\bar{\i}+(-2,-2))$ belongs to $\{E_j\}$, in which case the light gray square in Fig.~\ref{Fig5}{(}c) does not belong to $E$, or 
$S(\bar{\i}+(-1,-4))$ belongs to $\{E_j\}$, in which case the light gray square in Fig.~\ref{Fig5}{(}d) does not belong to $E$. Note that in the latter case we reach a 
contradiction if the light gray square in Fig.~\ref{Fig5}{(}d) also belongs to
$Q_{2k}(n)$. To this end we use the assumption $k>3$.

\smallskip
{\it Step $3$}. We can reason symmetrically if $E_{\bar{\j}}$ is a translation of $S$, say $S(\bar{\j}_1,\bar{\j}_2)$.

From Steps 1 and 2 we deduce that if $R(\bar{\j}_1,\bar{\j}_2)\subset Q_{2k}(n)$
is part of the family $\{E_j\}$ then all the translations $R(\bar{\j}_1+t,\bar{\j}_2-t)$ 
contained in $Q_{2k}(n)$ with $t\in\Z$ belong to the family $\{E_j\}$, and symmetrically that if
$S(\bar{\j}_1,\bar{\j}_2)\subset Q_{2k}(n)$
is part of the family $\{E_j\}$ then all the translations $S(\bar{\j}_1+t,\bar{\j}_2+t)$ 
contained in $Q_{2k}(n)$ 
with $t\in\Z$ belong to the family $\{E_j\}$

\smallskip
{\it Step $4$}. Consider now a set $E_j$ with $E_j\cap Q_{2k-4}(n)\neq \emptyset$. 
We may suppose again $E_j=R(\bar\j)$. From the previous steps also the sets
$R(\bar{\j}_1+t,\bar{\j}_2-t)$ intersecting $Q_{2k-4}(n)$ with $t\in\Z$ belong to the family $\{E_j\}$. 
Consider a unit square $Q$ in $Q_{2k-4}(n)$ neighbouring some of those $R(\bar{\j}_1+t,\bar{\j}_2-t)$. If it belonged to some $S(\bar\i)$ belonging to the family $\{E_j\}$ then by the previous steps the set $S(\bar\i+(-1,-1))$ would belong to the family $\{E_j\}$ (if $Q$ lies above some $R(\bar{\j}_1+t,\bar{\j}_2-t)$) 
or the set $S(\bar\i+(1,1))$ would belong to the family $\{E_j\}$ (if $Q$ lies below some $R(\bar{\j}_1+t,\bar{\j}_2-t)$).
In any case we would have a non-empty intersection between two elements of $\{E_j\}$,
which is a contradiction.
This implies that each such $Q$ belongs to a set $R(\bar\i)$ of the same modulated phase of $R(\bar\j)$. This gives that the two stripes neighboring the one of $R(\bar\j)$
are of the same modulated phase. Proceeding by finite induction we conclude that
all $E_j$ intersecting $Q_{2k-4}(n)$  belong to the same modulated phase.
\end{proof}

\begin{remark}\rm
It can be proved that the thesis of Lemma \ref{quadrato-interno} holds with $E_j\cap Q_{2k{-2}}(n) \neq \emptyset$. However the proof is more involved and we will not need such a sharp description.
\end{remark}

From Lemma \ref{quadrato-interno} we deduce that it is not possible to tessellate 
$\R^2$ using disjoint translations of both $R$ and $S$, or of different modulated phases 
of the same pattern, as stated in the following corollary.
\begin{corollary}\label{conf-periodiche}
Let $E = \bigcup E_j\in\A$ and suppose that $E = \R^2$.
Then there exists  $i\in\{1,\dots, 8\}$ such that $E_j\in \A_i$ for all $j$, or, equivalently,
$\{E_j\}=\ZZ_i$.
\end{corollary}

\begin{remark}[other types of chiral molecules]\rm
\begin{figure}[ht]
\centering
\includegraphics[width=.5\textwidth]{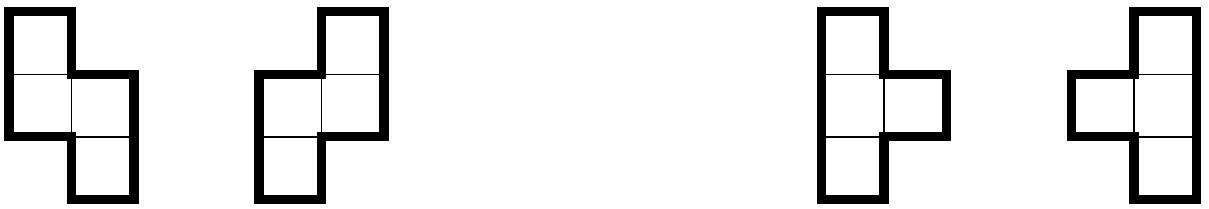}
\caption{Other types of chiral molecules}
\label{more-mol}
\end{figure}
In \cite{SW} other pairs of chiral molecules occupying four sites of a square lattice
are considered.
Such pairs, represented as union of squares, are pictured in Fig.~\ref{more-mol}.
\begin{figure}[ht]
\centering
\includegraphics[width=.85\textwidth]{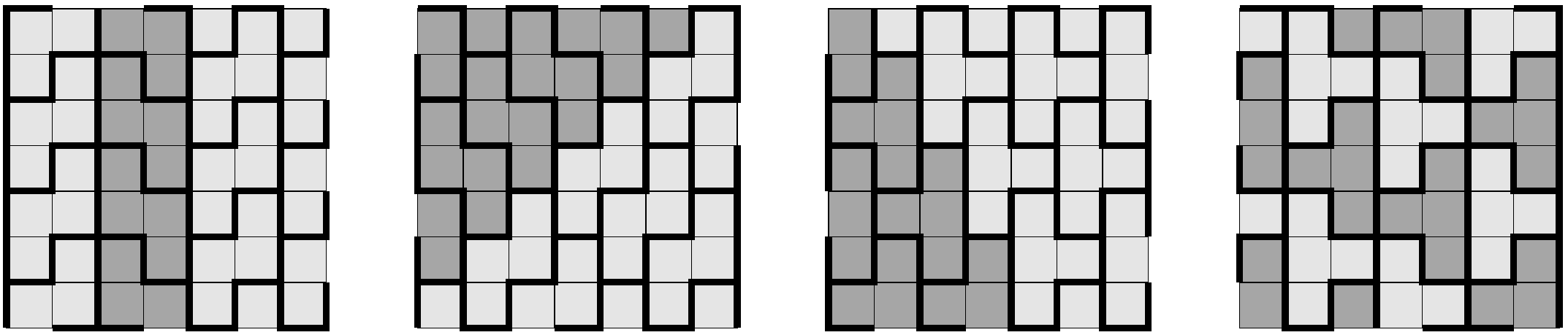}
\caption{Configurations violating Lemma \ref{quadrato-interno}}
\label{no-lemma}
\end{figure}The energy per molecule is again given by \eqref{enSW}, but in this case 
Lemma~\ref{quadrato-interno} does not hold, as shown by the configurations in 
Fig.~\ref{no-lemma}. As a result we do not have a parameterization of ground states
that can be used to define a compact convergence as in the next section.
\end{remark}

\section{Convergence to a partition}\label{Sect2}
Let $\Om\subset \R^2$ be an open bounded set with Lipschitz-continuous boundary. 
For each $\e>0$, we define ${\mathcal E}^\e(\Omega)$ as the collection of families of 
essentially disjoint unions of translations of $\e R$ and $\e S$ defined as
$$
{\mathcal E}^\e(\Omega):= 
\left\{ 
\{E^\e_j \}_j: 
\hbox{there exists } \{E_j \}_j\in {\mathcal E}\hbox{ such that }  E^\e_j =\e E_j
 \hbox{ for all }j \hbox{, and } E^\e_j\subset\Omega
\right\}.
$$
We denote by $\A^\e(\Omega)$ as the family of sets defined as unions of families in $\mathcal E^\e(\Omega)$; i.e., 
$$
\A^\e(\Omega):= 
\Biggl\{ 
\bigcup_j E_j : 
\{E_j\}_j\in {\mathcal E}_\e(\Omega)
\Biggr\}.
$$

We also define 
\bes
\ZZ_i^\e:= 
\begin{cases}
\{  \e R(n) : n\in\Z^2, n_2 + n_1 \equiv i \, \mod 4 \}   \text{ for }   i=1,\dots, 4,\\
 \{  \e S(n):  n\in\Z^2, n_2 - n_1 \equiv i \, \mod 4 \}    \text{ for } i=5,\dots, 8.
\end{cases}
\ees
Furthermore, we denote by $\P(\Om)$ the family of ordered partitions $(A_0,A_1,\dots, A_8)$ of 
$\Om$ into nine sets of finite perimeter.
%
\begin{definition}\label{def-convergenza}
We say that a sequence $ \{E^\e_j\}  \in \E^\e(\Omega)$ 
converges to $A\in \P(\Om)$, and we write $\{E^\e_j\} \to A$, if 
\begin{align}
\label{seconda} & |{\Om_i^\e}\triangle {A_i}|\to 0 \text{ as }\e\to 0,
\end{align}
where $\Om_i^\e := \bigcup \{ E_j^\e \colon  E_j^\e \in \ZZ_i^\e \}$ for each $i=1,\dots,8$, 
and $\Om_0^\e := \Om\setminus \big(\bigcup_{i=1}^8 \Om_i^\e \big)$.
\end{definition}

This notion of convergence is justified by the following compactness result.

\begin{theorem}[compactness]\label{te-co}
Assume that  $\{E_j^\e\}_j \in \E^\e(\Omega)$ is  such that, 
having set $E^\e= \bigcup_jE_j^\e$, we have 
\begin{equation}\label{equibounded}
C=\sup_\e \H^1 \big(\Om\cap\partial E^\e \big)<+\infty. 
\end{equation}
Then (up to relabeling) there exists a subsequence 
$\{E^\e_j\}$ converging to some $A\in\P(\Om)$
in the sense of Definition {\rm\ref{def-convergenza}}.
\end{theorem}
\begin{proof}
We use the notation
$Q^\e_{r}(x)$ for the rescaled cube $\e Q_{r}(x)= Q_{\e r}(\e x)$. 

Introduce a cover $\{Q^\e_{12}(n)\}_{n\in 4\Z^2}$ of $\R^2$ using squares of sides $12\e$ and center a point 
$\e n\in  4\e \Z^2$. 

Set 
\begin{align*}
& I^\e:=\{  n\in 4\Z^2: Q^\e_{12(n)}\cap\Om\neq\emptyset \},\\
&\hat{I}^\e:=\left\{  n\in  I^\e:  \H^1(Q^\e_{12}(n) \cap \partial E^\e)  \ge \e  \right\} \cup 
\{ n\in  I^\e:  Q^\e_{12}(n)\cap\partial\Om\neq\emptyset\}.
\end{align*}

From \eqref{equibounded} it follows that
\begin{equation}\label{stir}
\sharp \,  \hat{I}^\e\le {1\over\e}\sum_{n\in \hat{I}^\e}
 \H^1(Q^\e_{12}(n) \cap \partial E^\e)
 \le {9\over\e} \sum_{m\in 4 \Z^2}
 \H^1(\overline{Q^\e_{4}(m)} \cap \partial E^\e)
 \le 18 \frac{C}{\e},
\end{equation}
where the factor $18$ derives from the fact that a square $Q^\e_{4}(m)$
is contained in nine squares $Q^\e_{12}(n)$, and those squares have parts of
the boundary in common in pairs, so that the boundary of $\partial E^\e$ may be accounted for
twice in the last inequality.
As a consequence, if we denote
$$
\hat\Omega^\e=\bigcup_{n\in \hat{I}^\e} Q^\e_{12}(n),
$$
then
we have
\be\label{area-perimetro}
\bigl|\hat\Omega^\e\bigr|\le 144\,\e^2\cdot 18 \frac{C}{\e}= O(\e),
\ee
so that this set is negligible as the convergence of $\{E^\e_j\}_j$ is concerned.

For each $n\in I^\e\setminus\hat{I}^\e$ we apply Lemma \ref{quadrato-interno} to $Q^\e_{12}(n)$ and deduce that the corresponding $Q^\e_{4}(n)$ satisfies:
either 
\be\label{insieme1}
E^\e\cap Q^\e_{4}(n) = \emptyset,
\ee
or there exists  $i_\e(n) \in\{1,\dots, 8\}$ such that 
\be\label{insieme2}
E_j^\e\in \ZZ^\e_{i_\e(n)}
\ee
for each $j$ such that $E_j^\e\cap  Q^\e_{4}(n) \neq \emptyset$.
For all $i\in\{1,\dots, 8\}$ we define
$$
A^\e_i=\bigcup\Bigl\{ Q^\e_{4}(n): n\in I^\e_i\},\hbox{ where }  I^\e_i=\{n\in I^\e\setminus\hat{I}^\e, i_\e(n)=i\}
$$
and $i_\e(n)$ is defined in (\ref{insieme2}).
We can estimate the length of $\partial A^\e_i$ by counting the number of
the cubes of which it is composed which have a side neighboring a cube in the complement times the length $4\e$ of the corresponding interface, as
\begin{eqnarray*}
\H^1(\partial A^\e_i)&\le &4\e\sharp\{ n\in I^\e_i: \hbox{ there exists }n'\in 4\Z^2: |n-n'|=4
\hbox{ and } n'\not\in I^\e_i\}\\
&\le&
4\e\sharp\{ n\in I^\e_i: \hbox{ there exists }n'\in 4\Z^2: |n-n'|=4
\hbox{ and } n' \in \hat I^\e\}.
\end{eqnarray*}
Indeed if $n' \not\in \hat I^\e$ then by Lemma \ref{quadrato-interno} we would deduce
that $E^\e_j \in \ZZ_{i_\e(n')}$ for all $j$ such that $E_j^\e\subset Q^\e_{8}(n')$
and in particular for some  $j$ such that $E_j^\e\subset  Q^\e_{8}(n')$ and 
$E_j^\e\cap  Q^\e_{4}(n) \neq \emptyset$, which then implies that $i_\e(n')=i_\e(n)$.
We can then conclude our estimate using (\ref{stir}), and deduce (since the
number of possible neighbours $n'$ is $4$) that
\begin{eqnarray*}
\sup_\e\H^1(\partial A^\e_i)\le \sup_\e16\e\,\sharp\hat I^\e\le \sup_\e16\e\,{18 C\over\e}
<+\infty.
\end{eqnarray*}
By the compactness of (bounded) sequences of equi-bounded perimeter,
we deduce that there exist sets of finite perimeter $A_1,\ldots, A_8$ such that
(\ref{seconda}) holds for $i=1,\ldots, 8$, since 
$$
\Omega^\e_i\setminus A^\e_i\subset \hat\Omega_\e,
$$
which we already proved to be negligible. We finally deduce that (\ref{seconda}) holds for $i=0$ by the convergence of the complement of $\Omega^\e_0$.
\end{proof}
%
%
\section{Asymptotic analysis}\label{Sect3}
We now describe the asymptotic behaviour of perimeter energies defined on families
of mole\-cules. We will treat in detail a fundamental case, highlighting possible extensions and variations in the sequel.

For all $\{E_j\}\in  \E^\e$ we set
\begin{equation}\label{fhome}
\F^\e(\{E_j\},\Om) := \H^1 \big(\Om\cap\partial E \big),\hbox{ where }E=\bigcup_j E_j.
\end{equation}
We will prove that the asymptotic behaviour of $\F^\e$ as $\e\to0$ is described by an interfacial energy defined on partitions parameterized by the nine ground states described above. To that end, we first give a definition of the limit interfacial energy density by means of an asymptotic homogenization formula.

\subsection{Definition of the energy densities}
Given a unit vector $\nu\in S^1$ and $i,j\in\{0,1,\ldots,8\}$ with $i\neq j$, we define 
the family $\{E_h^{i,j,\nu}\}$ as follows. If $i\neq 0$ then
\begin{equation}\label{Ei0nu}
\{E_h^{i,0,\nu}\}=\{E\in\ZZ_i: E\cap\{x: x\cdot \nu>2\}\neq\emptyset\},
\end{equation}
\begin{equation}\label{E0jnu}
\{E_h^{0,i,\nu}\}= \{E_h^{i,0,-\nu}\}=\{E\in\ZZ_i: E\cap\{x: x\cdot \nu<-2\}\neq\emptyset\};
\end{equation}
i.e., $\{E_h^{i,0,\nu}\}$ is the family composed of elements of $\ZZ_i$ internal to
the half-plane $\{x: x\cdot \nu>0\}$, and symmetrically $\{E_h^{0,i,\nu}\}$ is the family 
composed of elements of $\ZZ_i$ internal to the half-plane $\{x: x\cdot \nu<0\}$.
If $i\neq0$ and $j\neq 0$ then 
\begin{equation}\label{Eijnu}
\{E_h^{i,j,\nu}\}= \{E_h^{i,0,\nu}\}\cup \{E_h^{0,j,\nu}\}.
\end{equation}
In this way we have defined the family $\{E_h^{i,j,\nu}\}$ for all $i\neq j\in\{0,\ldots,8\}$.
Note that $\{E_h^{j,i,\nu}\}=\{E_h^{i,j,-\nu}\}$.

The families defined above will allow to give a notion of boundary datum for 
minimum-interface problems on invading cubes.
More precisely, for all $T>0$ we define the neighbourhood of $\partial Q_T$
\begin{equation}\label{bordo} 
\delta Q_T := \{ x\in Q_T \colon \dist(x; \partial Q_T)< 4 \}.
\end{equation}
If $i\neq j$ then we set
\begin{eqnarray}\label{aij}
a_{i,j}(T,\nu)&:=&
 \min\Big\{
\F^1(\{E_k\},Q_{T})  \colon E_k\in\{E_h^{i,j,\nu}\} \hbox{ if } E_k\cap \delta Q_T\neq\emptyset
\Big\},
\end{eqnarray}
where $\F^1$ is defined in (\ref{fhome}) with $\e=1$. 

\begin{figure}[ht]
\centering
\includegraphics[width=.30\textwidth]{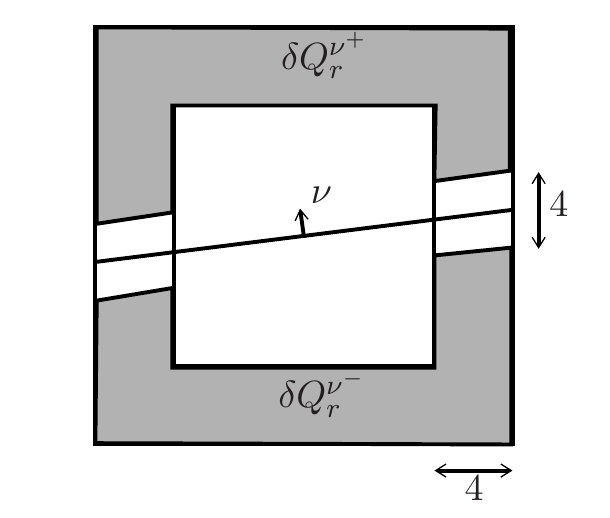}
\caption{Sets used in the definition of boundary conditions}
\label{Fig6}
\end{figure}
In order to illustrate this minimum problem, we refer to Fig.~\ref{Fig6},
and we denote by $Q^{\nu^+}_T$ and $Q^{\nu^-}_T$ the subsets 
$Q_T\cap\{x\cdot \nu >0\}$ and $Q_T\cap\{x\cdot \nu<0\}$ respectively.
Furthermore we set
\begin{align}
\label{bordo-sup} & \delta Q_T^{\nu^+} := \{ x\in Q^{\nu^+}_T \cap \delta Q_r \colon  x\cdot\nu > 2\},\\ 
\label{bordo-inf} & \delta Q_T^{\nu^-} := \{ x\in Q^{\nu^-}_T \cap \delta Q_r \colon x\cdot\nu<- 2 \},
\end{align}
The value $a_{i,j}(T,\nu)$ is the minimal
length of $Q_T\cap\partial E$, among $E\in\A$ obtained from a family coinciding with
$\{E_h^{i,j,\nu}\}$ on sets intersecting $\delta Q_T$. In particular, if $i,j\neq 0$, then
the sets $\delta Q_r^{\nu^+}$ and  $\delta Q_r^{\nu^-}$,
represented in Fig.~\ref{Fig6} by the shaded area, are covered by elements of $\{E_h\}$
in $\ZZ_i$ and $\ZZ_j$, respectively.



\smallskip
\begin{definition}[energy density]
The {\em surface energy density}  $f:\{0,\ldots,8\}\times \{0,\ldots,8\}\times S^1\to (0,+\infty)$
is defined by setting $f(i,i,\nu)=0$ and, if $i\neq j$, 
\be\label{energy-density}
f (i,j,\nu):= \liminf_{T\to+\infty} 
\big(|\nu_1| \vee |\nu_2| \big)
\frac{a_{i,j}(T,\nu)}{T},
\ee
where the $a_{i,j}$ are defined by minimization on $Q_T$ in {\rm(\ref{aij})}.
\end{definition}

The normalization factor $|\nu_1| \vee |\nu_2|$ takes into account the length of $Q_1\cap\{x: x\cdot\nu=0\}$.

%

\begin{remark}\rm
1. ({\em symmetry})
Note that the symmetric definition of $\{E_h^{i,j,\nu}\}$ gives that $$a_{i,j}(T,\nu)=a_{j,i}(T,-\nu)$$ for all $i,j\in\{0,\ldots,8\}$ and $\nu\in S^1$, so that
$$
f(i,j,\nu)=f(j,i,-\nu),
$$
which is a necessary condition for a good definition of a surface energy.

2. ({\em continuity}) For all $i,j$ the function $f(i,j,\nu)$ is continuous on $S^1$. In order to check this, given $\nu$ and $\nu'$, 
for fixed $T$ let $\{E^{\nu,T}_h\}$ be a minimizer for $a_{i,j}(T,\nu)$,
and let
$$
E^{\nu,T}=\bigcup\{E^{\nu,T}_h: |E^{\nu,T}_h\cap Q_T|\neq0\}
$$
We define $\{E^{\nu',T}_h\}$ as 
$$
\{E^{\nu,T}_h: |E^{\nu,T}_h\cap Q_T|\neq0\}\cup\{E_h^{i,j,\nu'}: |E_h^{i,j,\nu'}\cap E^{\nu,T}|=0\}
$$
and use it to test $a_{i,j}(T+8,\nu')$, for which it is an admissible test family. We then get
$$
a_{i,j}(T+8,\nu')\le a_{i,j}(T,\nu)+ C + CT|\nu-\nu'|,
$$
the constant $C$ estimating the contribution on $Q_{T+4}\setminus Q_T$, and the last term
due to the mismatch of the boundary conditions close to $\partial Q_T$. 
From this inequality
we deduce that $f(i,j,\nu')- f(i,j,\nu)\le C|\nu-\nu'|$ and, arguing symmetrically, that
$$
|f(i,j,\nu')- f(i,j,\nu)|\le C|\nu-\nu'|.
$$
\end{remark}

\begin{remark}\rm
1. The liminf in (\ref{energy-density}) is actually a limit. This can be proved directly by a subadditive argument,
or as a consequence of the property of convergence of minima of $\Gamma$-convergence (see Remark \ref{co-min}).

2. An alternate formula can be obtained by defining $Q^\nu$ as the unit square centered in $0$ and with one side orthogonal to $\nu$. We then have
\be\label{energy-density-2}
f (i,j,\nu):= \liminf_{T\to+\infty} 
{1\over T} \min\Big\{
\F^1(\{E_k\},TQ^\nu)  \colon E_k\in\{E_h^{i,j,\nu}\} \hbox{ if } E_k\cap \delta TQ^\nu\neq\emptyset
\Big\},
\ee
where again $\delta TQ^\nu=\{ x\in TQ^\nu \colon \dist(x; \partial Q_T)< 4 \}$.
Note that in this case we do not need to normalize by $|\nu_1| \vee |\nu_2|$ since the length of $Q^\nu\cap\{x: x\cdot\nu=0\}$ is $1$.

This formula can be again obtained as a consequence of the $\Gamma$-convergence
Theorem. Conversely, a proof of Theorem \ref{thm:gammalim} 
using this formula can be obtained following the same 
line as with the first formula, but is a little formally more complex due to the fact that the sides of $Q^\nu$ are not oriented in the coordinate directions. The changes in the proof can be found in the paper by Braides and Cicalese \cite{BC-last}. Note that usually extensions to dimensions higher than two are easier with this second formula.

\end{remark}

\subsection{$\Gamma$-limit}
Let  $\F_\Omega^\e$ be the functional defined for each $\{E_j\}\in  \E^\e (\Omega)$ as
\begin{equation}\label{gamma-lime}
\F_\Omega^\e(\{E_j\})=\begin{cases} \F^\e(\{E_j\},\Omega) & \hbox{ if $\{E_j\}\in  \E^\e (\Omega)$ }
\cr+\infty &\hbox{ otherwise.}\end{cases}
\end{equation}

We introduce the functional that assigns to every partition $A=$ $\{A_0,\dots,A_8  \}\in\P(\Om)$ 
the real number 
\be\label{gamma-lim}
\F_\Omega(A):= \sum_{i=0}^7\sum_{j=i+1
}^8\int_{\Omega\cap\partial A_i\cap\partial A_j} f (i,j,\nu^i)d\H^{1}
+
 \sum_{i=1}^8\int_{\partial A_i\cap\partial \Om} f (i,0,\nu^i)d\H^{1},
\ee
where $\nu^i$ is the {\em inner normal} of the set $E_i$ and $f$ is the interface energy defined 
above. We use the notation $\partial A$ to denote the {\em reduced boundary} of a set
of finite perimeter $A$. Since we consider topological boundaries which coincide $\H^1$ almost everywhere with the corresponding reduced boundaries this notation will not cause confusion.

We use the same notation in \eqref{gamma-lim} also when $\Omega$ is not bounded.
In particular we can consider $\Omega=\R^2$, in which case the last surface
integral is not present. In that case we use the notation $\F(A)$ in the place of $\F_{\R^2}(A)$.

We then have the following result.

\begin{theorem}\label{thm:gammalim}
Let $\Om\subset \R^2$ be an open bounded set with Lipschitz-continuous boundary. 
The sequence of functionals $\{\F^{\e}_\Omega\}$ defined in \eqref{gamma-lime}
$\Gamma$-converges, as $\e\to 0^{+}$, to the functional $\F_\Omega$ defined by \eqref{gamma-lim}, with respect to 
the convergence of Definition {\rm\ref{def-convergenza}}. 
\end{theorem}

\begin{remark}[$BV$-ellipticity]\rm
As a consequence of the lower semicontinuity of $\F$ we obtain that $f$ is {\em$BV$-elliptic} \cite{AB}.
In particular the extension by one-homogeneity of $f(i,j,\cdot)$ is convex for all $i,j$ and we have
the subadditivity property $f(i,j,\nu)\le f(i,k,\nu)+f(k,j,\nu)$.
\end{remark}

\begin{proof}

\noindent {\em Lower bound.}\\
We consider a partition $A=(A_0,\ldots, A_8)\in\P(\Omega)$ and a family $\{E^\e_h\}$ converging to $A$.
We can suppose that $\liminf_{\e\to0}\F_\Omega^\e(\{E^\e_j\})<+\infty$. We choose a subsequence 
$(\e_k)$ such that $$\lim_k\F_\Omega^{\e_k}(\{E^{\e_k}_j\})=\liminf_{\e\to0}\F_\Omega^\e(\{E^\e_j\})$$
and such that the measures on $\Omega$ defined by $\mu_{\e_k}(B)=\F_\Omega^{\e_k}(\{E^{\e_k}_j\},B)$ weakly$^*$ converge to some measure $\mu$. In order not to overburden the notation we denote $\e_k$ simply by $\e$.

We use the blow-up method of Fonseca and M\"uller \cite{FM}, which consists in giving a lower bound of the
density of the measure $\mu$ with respect to the target measure $\H^1$ restricted to $\bigcup_i\partial A_i$.
We refer to \cite{BMS} for technical details regarding the adaptation of this method to homogenization problems.

In the present case the blow-up is performed at $\H^1$-almost
every point $x_0$ in $\bigcup_i\partial A_i$. Note that this comprises also the points in $\partial \Omega$ 
where the inner trace of the partition at that point is not the set $A_0$. 
By a translation and slight adjustment argument 
(due to the fact that in general $x_0\not\in\e\Z^2$) we can simplify our notation by supposing that $x_0=0$.
It then suffices to show that 
\be\label{derivative}
\lim_{\rho\to 0}\lim_{\e\to 0} \frac{\F^\e( \{E_h^\e\},Q_\rho)}
{\H^1(\partial A_i\cap\partial A_j\cap Q_\rho)}
\geq
f (i,j,\nu),
\ee
%
supposing that $0\in\partial A_i\cap\partial A_j$ and setting $\nu=\nu^i(0)$.

Let $\e=\e(\rho)$ be such that $T=\rho/\e\to+\infty$ and 
$$
\lim_{\rho\to 0}\lim_{\e\to 0} \frac{\F^\e( \{E_h^\e\},Q_\rho)}
{\H^1(\partial A_i\cap\partial A_j\cap Q_\rho)}=
\lim_{\rho\to 0} \frac{\F^{\e(\rho)}( \{E_h^{\e(\rho)}\},Q_\rho)}
{\H^1(\partial A_i\cap\partial A_j\cap Q_\rho)}.
$$

Note that (by definition of reduced boundary) ${1\over\rho}A \cap Q_1$ 
tends to $A^{i,j,\nu}\in \P(Q_1)$, 
where $A^{i,j,\nu}_i= Q_1^{\nu^+}$ and  $A^{i,j,\nu}_j = Q_1^{\nu^-}$ (and therefore 
$A^{i,j,\nu}_h =\emptyset $ for each $h\neq i,j$  and the interface $\partial A^{i,j,\nu}_i\cap\partial A^{i,j,\nu}_j$ is the segment $Q_T\cap \{x\cdot\nu=0\})$.

In order to prove \eqref{derivative} it is sufficient to show that 
\be\label{lower-bound}
\liminf_{T\to \infty}\frac{|\nu_1| \vee |\nu_2|}{T}\F^1(\{E^T_h\},Q_T) \geq
f (i,j,\nu),
\ee
with
$$
\{E^T_h\}=\Bigl\{{1\over\e}E^\e_h\Bigr\}.
$$

We then define $\omega^\pm_T$ as follows. We set
$$
\omega^+_T=\bigcup_h\{ E^T_h: E^T_h\in\ZZ_i,\ E^T_h\subset Q_T^{\nu+}\} 
$$
if $i>0$ and 
$$
\omega^+_T=Q_T^{\nu+}\setminus \bigcup_h\{ E^T_h\} 
$$
if $i=0$, and similarly
$$
\omega^-_T=\bigcup_h\{ E^T_h: E^T_h\in\ZZ_j,\ E^T_h\subset Q_T^{\nu-}\} 
$$
if $j>0$ and 
$$
\omega^-_T=Q_T^{\nu-}\setminus \bigcup_h\{ E^T_h\} 
$$
if $j=0$.
Then 
\begin{align}
\label{small}  |Q_T \setminus (\omega^+_T \cup \omega^-_T) | = o(T^2).
\end{align}


%

We now show that, up to a small error, $\{E^T_h\}$ can be modified in order to fulfill  the boundary conditions defined in \eqref{aij} for each $T$. From this \eqref{lower-bound} follows by the definition of $a_{i,j}$.  

Let $\sigma \ll 1$ and $k_T:= [\sigma T/12] $. We introduce a partition of the frame 
$$
\Bigl\{(x_1,x_2): {T\over 2}(1- \sigma)<  |x_i|  < {T\over 2} , i=1,2\Bigr\}= Q_T \setminus \overline{Q_{(1-\sigma)T}}
$$
into $k_T$ subframes  $C^T_n$ of thickness $\sigma T / 2 k_T$:
$$
C^T_n:=\Big\{ 
(x_1,x_2):   {T\over 2}\Bigl(1- \sigma +  \frac{ (n - 1) \sigma}{ k_T}\Bigr)<   |x_i|  <    {T\over 2}\Bigl(1- \sigma +  \frac{n\sigma }{ k_T}\Bigr), i=1,2
\Big\}$$ 
for $n=1,\dots, k_T$. Note that $\sigma T / 2 k_T\ge 6$.

From \eqref{small} it follows that  there exists $n_T \in\{ 1, \dots, k_T\}$ such that 
\be\label{cornice-i}
 |C_{n_T}^T \setminus (\omega^+_T \cup \omega^-_T) |   =\frac{o(T^2)}{k_T} = \frac{o(T)}{\sigma}.
\ee
%
We define
$$
R_T={T}\Bigl(1- \sigma +  \frac{n_T\sigma }{ k_T}\Bigr)
$$
so that 
$$
C_{n_T}^T= Q_{R_T}\setminus {\overline Q_{R_T-\frac{T\sigma }{ k_T}}}.
$$

We now construct a family $\{\tilde E^T_h\}$ satisfying the desired boundary conditions by 
taking the elements of the family $\{E_h^{i,j,\nu}\}$ defined in \eqref{Eijnu} which are not contained 
in $Q_{R_T}$ union those in $\{E^T_h\}$ which do not intersect any of the former. 
More precisely, we define
$$
E^{i,j,\nu}_{R_T}=\bigcup\Bigl\{ E_h^{i,j,\nu}: E_h^{i,j,\nu}\not\subset Q_{R_T}\Bigr\}
$$
and 
$$
\{\tilde E^T_h\}=\{ E^T_h: |E^{i,j,\nu}_{R_T}\cap E^T_h|=0\}\cup
\Bigl\{ E_h^{i,j,\nu}: E_h^{i,j,\nu}\not\subset Q_{R_T}\Bigr\}.
$$
Let $\tilde E^T=\bigcup_h \tilde E^T_h$.
Note that, up to $\H^1$-negligible sets
\begin{equation}\label{inclu}
(\partial \tilde E^T\setminus  \partial E^T)\cap \overline Q_{R_T}\subset
 (C_{n_T}^T \setminus (\omega^+_T \cup \omega^-_T)).
 \end{equation}
This inclusion is proved noting that points in the boundary of $\tilde E^T$ which are not in the boundary of $E^T$
can be subdivided into two sets: points that are in $E^{i,j,\nu}_{R_T}$ and those that are not.
The first ones must belong to some $E^{i,j,\nu}_h$  with 
$\H^1\big(\partial E^{i,j,\nu}_h\cap (C_{n_T}^T \setminus (\omega^+_T \cup \omega^-_T))\big)\neq0$, 
the second ones must be interior to $E^T$ but on the boundary of some $E^T_h$ with 
$|E^{i,j,\nu}_{R_T}\cap E^T_h|\neq 0$, so that in particular they also belong to  
$C_{n_T}^T \setminus (\omega^+_T \cup \omega^-_T)$. From \eqref{inclu}
and the fact that $C_{n_T}^T \setminus (\omega^+_T \cup \omega^-_T)$ is composed
of unit squares, using \eqref{cornice-i}, we have
\begin{equation}\label{inclu-dise}
\H^1\Bigl((\partial \tilde E^T\setminus  \partial E^T)\cap \overline Q_{R_T}\Bigr)= {o(T)\over\sigma}.
 \end{equation}

We can estimate
\begin{eqnarray}
\H^1(Q_T\cap \partial \tilde E^T)&=&
\H^1\Bigl(\overline Q_{R_T-\frac{T\sigma }{k_T}}\cap \partial \tilde E^T\Bigr)\\ \nonumber
&&+ \H^1\Bigl(\Bigl(Q_{R_T}\setminus \overline Q_{R_T-\frac{T\sigma }{k_T}}\Bigr)\cap \partial \tilde E^T\Bigr)
+\H^1((Q_T\setminus Q_{R_T})\cap \partial \tilde E^T)\\
&=&\nonumber
\H^1\Bigl(\overline Q_{R_T-\frac{T\sigma }{k_T}}\cap \partial  E^T\Bigr)
+
\H^1\Bigl(C^T_{n_T}\cap \partial \tilde E^T\cap\partial E^T\Bigr)\\ \nonumber
&&
+ \H^1\Bigl(C^T_{n_T}\cap \partial \tilde E^T\setminus\partial E^T\Bigr)
+\H^1((Q_T\setminus Q_{R_T})\cap \partial \tilde E^T)\\
&\le&\nonumber
\H^1\Bigl(\overline Q_{R_T}\cap\partial E^T\Bigr)\\ \nonumber
&&
+ \H^1\Bigl(C^T_{n_T}\cap \partial \tilde E^T\setminus\partial E^T\Bigr)
+\H^1((Q_T\setminus Q_{R_T})\cap \partial \tilde E^T)
\\
&\le&\nonumber
\H^1\Bigl(Q_{T}\cap\partial E^T\Bigr) +{o(T)\over\sigma}+ 4 T\sigma.
\end{eqnarray}
In terms of the functionals $\F^1$ this reads
$$
\F^1(\{\tilde E_h^T\}, Q_T)  \leq 
\F^1(\{E_h^T\}, Q_T)  +  \frac{o(T)}{\sigma} + 4 \sigma T,
$$
which in turn yields
\bes
\liminf_{T\to \infty}\frac{|\nu_1| \vee |\nu_2|}{T}\F^1(\{ E_h^T\}, Q_T) \geq
\liminf_{T\to \infty}\frac{|\nu_1| \vee |\nu_2|}{T} \F^1(\{\tilde E_h^T\}, Q_T)
-4\sigma.
\ees
 Estimate \eqref{lower-bound} now follows from the arbitrariness of $\sigma$ and the definition of $\tilde E^T$.

Note that if the blow-up is performed at a point in $\partial\Omega$, then $\nu$ is the inner normal to $\Omega$,
$i\ge 1$ and $j=0$, which gives the boundary term in $\F_\Omega$.

\medskip
\noindent
{\em Upper bound.}\\
We need to show that for each $A\in\P(\Om)$ there exists a sequence $\{E^\e_j\}\in \E^\e(\Omega)$ converging to 
$A$ and such that
$\limsup_\e \F^\e(\{E^\e_j\},\Omega) \leq  \F_\Omega(A)$.
We can choose polyhedral sets $\Omega^h\subset\subset\Omega$ and
polyhedral partitions $A^h$ such that $|\Omega\setminus\Omega^h|\to 0$, $|A^h_i\triangle A_i|\to0$,
$\H^1(\partial\Omega^h)\to \H^1(\partial\Omega)$ and $\H^1(\partial A^h_i\cap\Omega^h)\to \H^1(\partial A_i\cap\Omega)$, $A^h_0\supset {\mathbb R}^2\setminus\Omega$, so that $\F(A^h)\to \F_\Omega (A)$
by the continuity of $f$.
The existence of such $\Omega^h$ follows from the regularity of $\Omega$, while 
the construction of partitions $A^h$ can be derived from \cite{CGM}.
By an usual approximation argument (\cite{GCB} Section 1.7)
it thus suffices to construct recovery sequences for $\F(A)$ in the case when  
$\bigcup_{i=1}^8A_i\subset\subset\Omega$ and each element of the partition is a polyhedral set, provided that the approximating $\{E^\e_j\}$ belongs to $\E^\e(\Omega)$. In other words, it suffices to construct recovery sequences for $\F(A)$
in the case when each $A_i$ is a bounded polyhedral set for $i\ge 1$, provided that the approximating $E^\e_j$ belong to a small neighbourhood of $\bigcup_{i=1}^8 A_i$.

Since we will reason locally, we exhibit our construction when the target partition is composed of the
two half-planes $A_l=\{x: \nu\cdot x>0\}$ and $A_k=\{x: \nu\cdot x<0\}$, and $\nu$ is a rational direction;
i.e., there exists $L\in\R$ such that $L\nu\in\Z^2$. We fix $\eta>0$ and $T=T_\eta$ such that  
\begin{equation}\label{defiti} 
\frac{|\nu_1| \vee |\nu_2|}{T} a_{l,k}(T, \nu)
\le f (l,k,\nu)+\eta
\end{equation}
Up to choosing a slightly larger $T$ (at most larger $4L$ than the previous one)
we can suppose that 
$$
{T\over {|\nu_1|\vee|\nu_2|}}\nu\in 4\Z^2. 
$$

Indeed, this amounts to an additional error proportional 
to $4L/(T+4L)$ in \eqref{defiti}, which we can include in $\eta$.

\begin{figure}[ht]
\centering
\includegraphics[width=.9\textwidth]{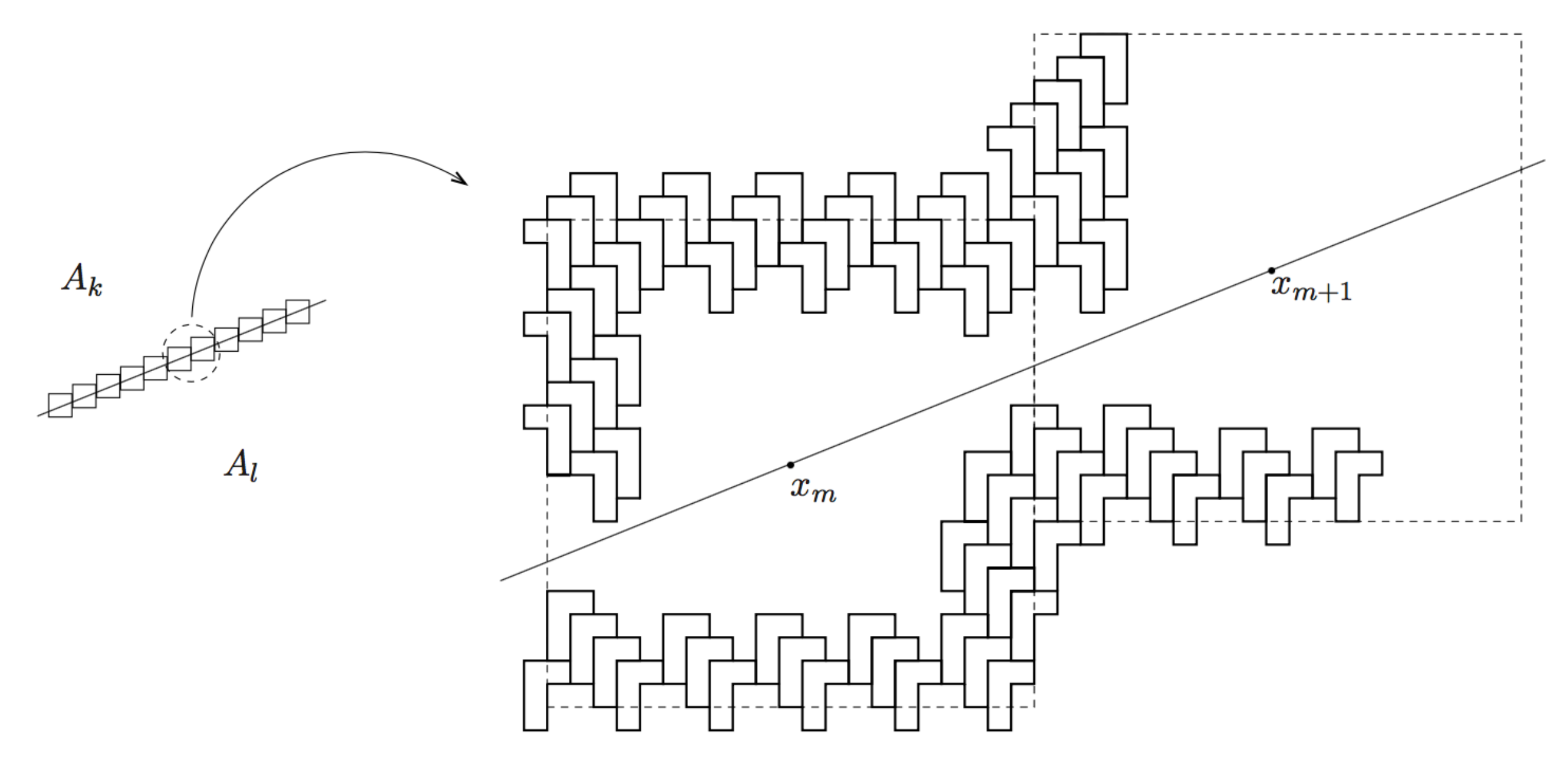}
\caption{Construction of a recovery sequence at an interface}
\label{fig:upper-bound}
\end{figure}
Let $\{E^T_j\}$ be a minimal family for $a_{l,k}(T, \nu)$.
We set $\nu^\perp=(-\nu_2,\nu_1)$.
We construct a sequence of molecules by covering the interfacial line $\{x: \nu\cdot x=0\}$
with the disjoint squares $E^T_j+mT{({|\nu_1|\vee|\nu_2|})^{-1}}\nu^\perp $ ($m\in\Z$), up to a discrete set of points,
and consider the optimal family inside each such square (see Fig.~\ref{fig:upper-bound}). Since the centres of such
cubes differ by a multiple of $4$ in each component, we can choose such optimal 
families as the translation of a single family, and match on the boundary of
each cube with elements of $\{E^{l,k,\nu}_j\}$, which allows to extend them outside
the union of the covering cubes. Note that this extension has zero energy, except
for $\nu={1\over\sqrt2}(\pm1,1)$, for which we may have a small contribution due to
a fixed number of molecules close to the vertices of the cubes on $\{x: \nu\cdot x=0\}$;
again this error will be taken care of by $\eta$.

We define $\{E^{l,k,T,\nu}_j\}$ as the union of all families 
\begin{equation}\label{fam-1}
\Bigl\{E^T_j+x_m: E^T_j\cap Q_T\neq \emptyset\Bigr\},\qquad x_m=m{T\over {|\nu_1|\vee|\nu_2|}}\nu^\perp 
\end{equation}
for $m\in\Z$, and the family
\begin{equation}\label{fam-2}
\Biggl\{E\in \{E^{l,k,\nu}_j\}: E\cap\bigcup_{m\in\Z} \Bigl(Q_T+m{T\over {|\nu_1|\vee|\nu_2|}}\nu^\perp\Bigr)=\emptyset\Biggr\}.
\end{equation}

Let now $\{E^\e_j\}$ be defined as $\{\e E_j\}$ with $\{E_j\}
=\{E^{l,k,T,\nu}_j\}$ the family just described.
Note that for all bounded open set $B$ with Lipschitz boundary 
such that $\H^1(\{ \nu\cdot x=0\}\cap \partial B)=0$
we have
$$
\limsup_{\e\to0} \F(\{E^\e_j\},B)\le \Bigl(f(l,k,\nu)+\eta\Bigr) \H^1(\{ \nu\cdot x=0\}\cap B).
$$

We now fix $A_i$ bounded polyhedral sets $i=1,\ldots, 8$, and 
repeat the construction described above close to each interface. 
To that end, we denote 
$$
\bigcup_{i=1}^8\partial A_i=\{p_m\},
$$ 
where $p_m$ are a finite number of segments with endpoints $x^+_m$ and $x^-_m$. 
Let $l(p_m)$ and $k(p_m)$ be the indices such that
$$
p_m\subset \partial A_{l(p_m)}\cap\partial A_{k(p_m)}
$$
and let $\nu(p_m)$ be the inner normal to $A_{l(p_m)}$ at $p_m$.
In our approximation argument it is not restrictive to suppose that 
$\nu(p_m)$ is a rational direction.
We fix $\eta$ and $T_m$ such that
\begin{equation}\label{defiti-m} 
\frac{|\nu_1(p_m)| \vee |\nu_2(p_m)|}{T_m} a_{l,k}(T_m, \nu(p_m))
\le f (l(p_m),k(p_m),\nu(p_m))+\eta
\end{equation}
and 
$$
{T_m\over {|\nu_1(p_m)|\vee|\nu_2(p_m)|}}\nu(p_m)\in 4\Z^2.
$$

We choose $M$ large enough so that the distance between all points
of 
$$
\bigcup\{p_m: x \hbox{ is an endpoint of }p_m\}\cap \partial Q_{\e M}(x)
$$
is larger than $2\e(4+\sup_m  T_m)$. 

Let $x^m_\e\in 4\e\Z^2$ be such that $p_m$ is contained in a tubular neighbourhood of the 
line through $x^m_\e$ and orthogonal to $\nu(p_m)$; i.e., such that
$$
p_m\subset \{x: (x-x^m_\e)\cdot\nu(p_m)=0\}+ Q_{4\e}.
$$
We denote 
\begin{eqnarray*}
C^\e_m&=&\bigcup\Bigl\{ x^m_\e+\e h{T_m\over {|\nu_1(p_m)|\vee|\nu_2(p_m)|}}\nu(p_m)^\perp+ Q_{\e T_m}: \\
&&\Bigl(x^m_\e+\e h{T_m\over {|\nu_1|\vee|\nu_2|}}\nu(p_m)^\perp+ Q_{\e T_m}\Bigr)
\cap p_m\neq\emptyset, h\in\Z\Bigr\}.
\end{eqnarray*}
Let $\{E^{\e,m}_j\}$ be the elements of the family $\{\e E^{l,k,T,\nu}_j+x^m_\e\}$
intersecting $C^\e_m\setminus \big(Q_{\e M}(x^+_m)\cup Q_{\e M}(x^-_m)\big)$,
where $\{E^{l,k,T,\nu}_j\}$ is constructed in \eqref{fam-1}--\eqref{fam-2} 
with $l=l(p_m)$, $k=k(p_m)$ and $\nu=\nu(p_m)$.
We then define $\{E^\e_j\}$ as the union of all $\{E^{\e,m}_j\}$
and of all the families
$$
\Bigl\{\e E_j: E_j\in\ZZ_i: E_j\subset A_i\setminus \bigcup_m 
\Bigl(C^\e_m\cup Q_{\e M}(x^+_m)\cup Q_{\e M}(x^-_m)\Bigr)\Bigr\}.
$$

Let $E^\e=\bigcup_j E^\e_j$. Note that the contributions due to the part of $\partial E^\e$ contained
in each set $Q_{\e M}(x^\pm_m)$ is at most of the order $\e M$.
We then have that $\{E^\e_j\}$ converges to $(A_0,\ldots, A_8)$ and 
\begin{eqnarray*}
\limsup_{\e\to 0} \F^\e(\{E^\e_j\})
&\le&\sum_m 
 \Bigl(f(l(p_m),k(p_m),\nu(p_m))+\eta\Bigr) \H^1(p_m)
\le \F(A)+C\eta.
 \end{eqnarray*}
By the arbitrariness of $\eta$ we obtain the upper bound.
\end{proof}

\begin{remark}\rm
The hypothesis that $\Omega$ be bounded can be removed.
In particular we can consider $\Omega=\R^2$, in which case the
term on the boundary of $\Omega$ in \eqref{gamma-lim} disappears.
The theorem can be proved in the same way, but the
notion of convergence must be slightly changed by requiring
that (\ref{seconda}) holds when restricted to bounded sets. 

On the other hand, we can define  $\G_\Omega^\e$ for $\{E_j\}\in  \E^\e $ as
\begin{equation}\label{gamma-lime-z}
\F_\Omega^\e(\{E_j\})=\begin{cases} \F^\e(\{E_j\},\Omega) & \hbox{ if $\{E_j\}\in  \E^\e$ }
\cr+\infty &\hbox{ otherwise;}\end{cases}
\end{equation}
i.e., we do not require the sets $\E_j$ to be contained in $\Omega$. The $\Gamma$-limit
is the same except for the boundary term on $\partial \Omega$, which again disappears.
The liminf inequality clearly holds in the same way, while a recovery sequence can be obtained 
by considering first target partitions that can be extended as sets of finite perimeter in an open 
neighbourhood of $\Omega$, and then argue by density.
\end{remark}

\begin{figure}[ht]
\centering
\includegraphics[width=.60\textwidth]{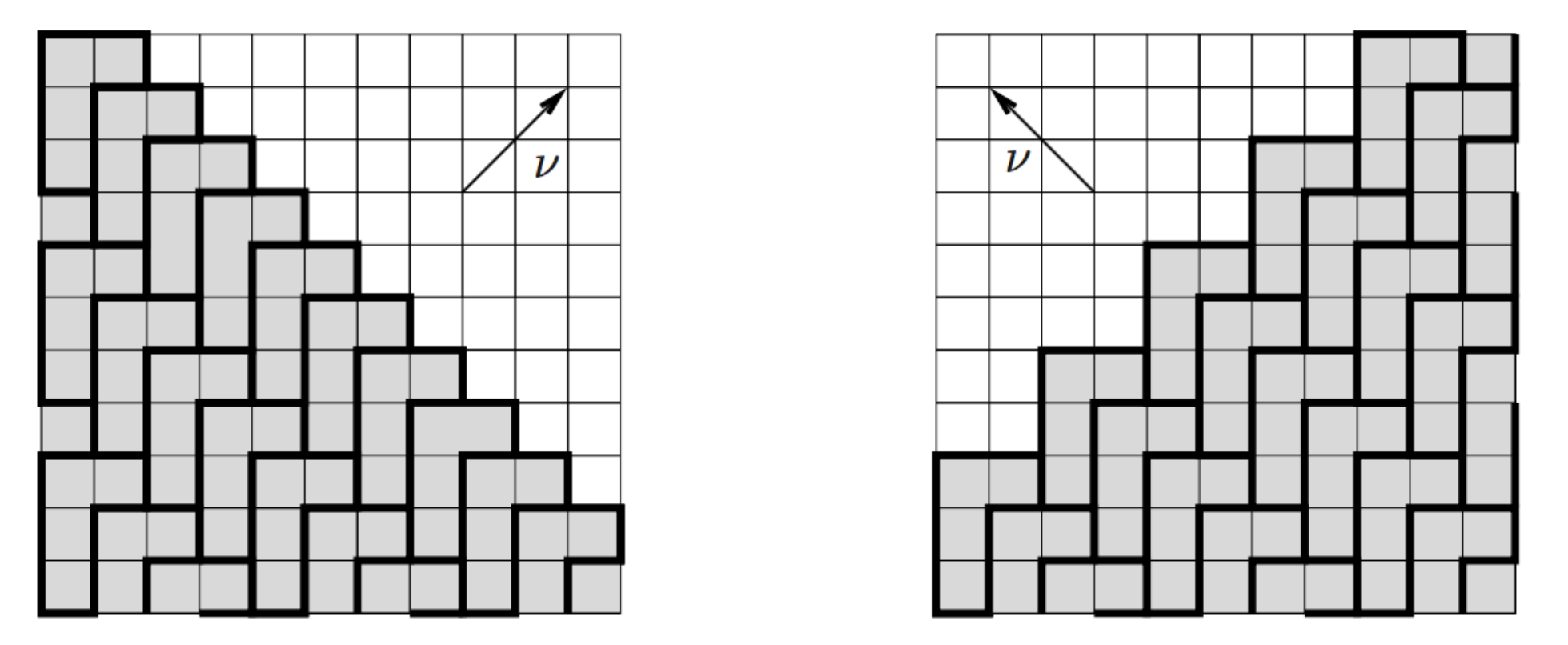}
\caption{Optimal sets when $j=0$ in diagonal directions}
\label{Fig11}
\end{figure}

\subsection{Description of $f$}\label{descrip} \ 

\smallskip\noindent{\em Computation of $f(i,0,\nu)$}.\ \  
Given any $i\in\{1,\ldots, 8\}$, we will explicitly compute $\varphi$,
the one-positively extension of $f(i,0,\cdot)$. Since this function turns out to be symmetric, we also have $\varphi(\nu)=f(0,i,\nu)$.
We treat in detail the case $i\in\{1,\ldots,4\}$. By a symmetry argument with respect to the
vertical direction, we have $f(i,0,\nu)=\varphi(-\nu_1,\nu_2)$ for $i\in\{5,\ldots,8\}$.

We preliminarily note that a lower bound for $\varphi$ is computed by removing the constraint that
the elements of $\{E_h\}$ be chiral molecules; i.e., taking $E_h$ unit squares in the lattice.
The computation for the $\Gamma$-limit without the constraint is simply $\|\nu\|_1=|\nu_1|+|\nu_2|$
(see \cite{ABC}), so that we have $\varphi(\nu)\ge|\nu_1|+|\nu_2|$.

We can check that we have equality for $\nu_1=\pm\nu_2$. Indeed, for such $\nu$ the
 optimal families are simply $\{E^{i,0,\nu}_h\}$, whose corresponding
 sets are those described in Fig.~\ref{Fig11}. 
 
\begin{figure}[ht]
\centering
\includegraphics[width=.6\textwidth]{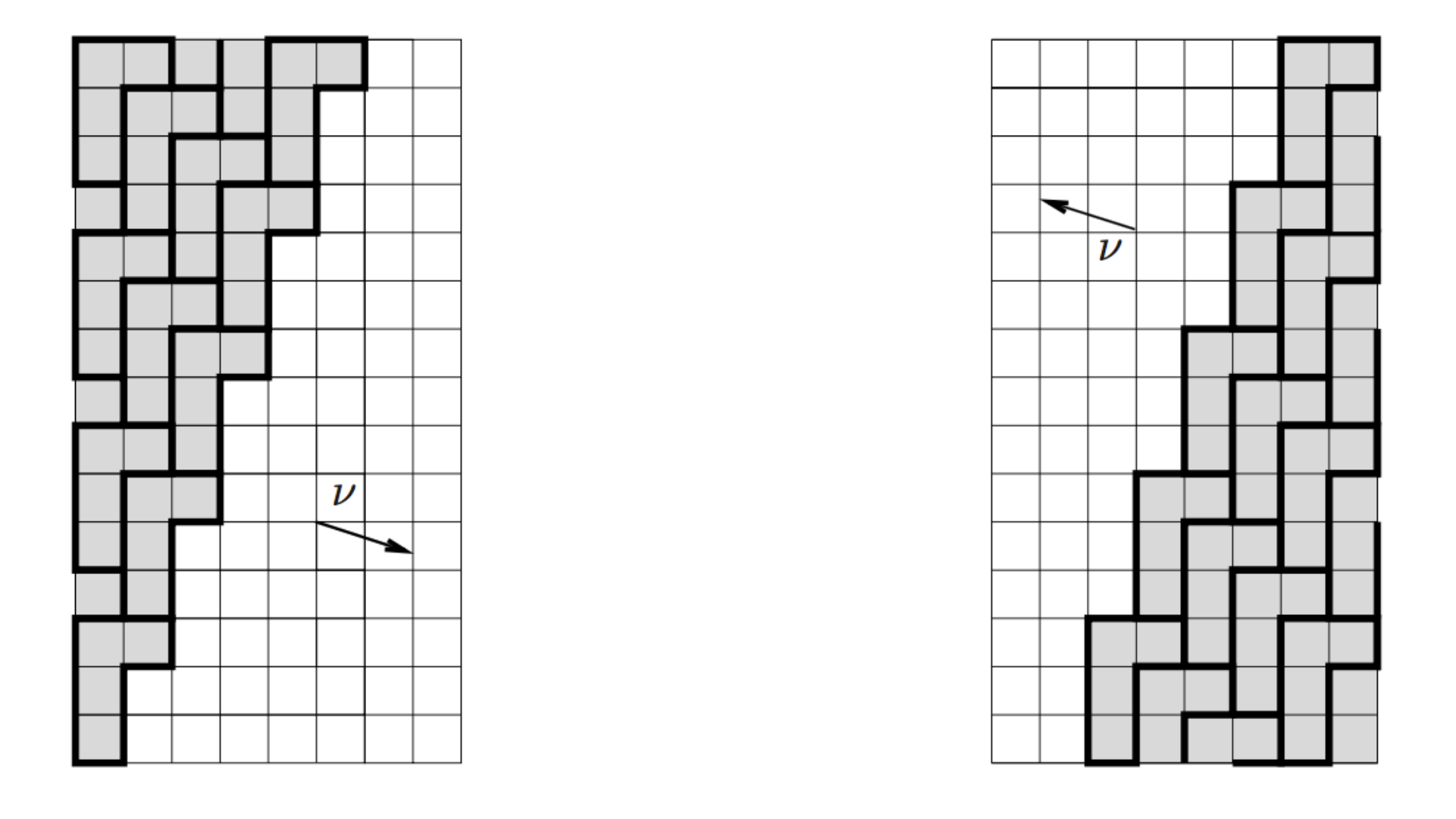}
\caption{Optimal sets when $j=0$ in direction $\pm{1\over\sqrt{10}}(3,-1)$}
\label{Fig21}
\end{figure}
Note that the value in the two direction is the same,
 but the `micro-geometry' of optimal sets is (slightly) different. Two other values in which we have 
 equality are with $\nu_1= 3\nu_2$, with optimal families pictured in Fig.~\ref{Fig21}.

\begin{figure}[ht]
\centering
\includegraphics[width=.45\textwidth]{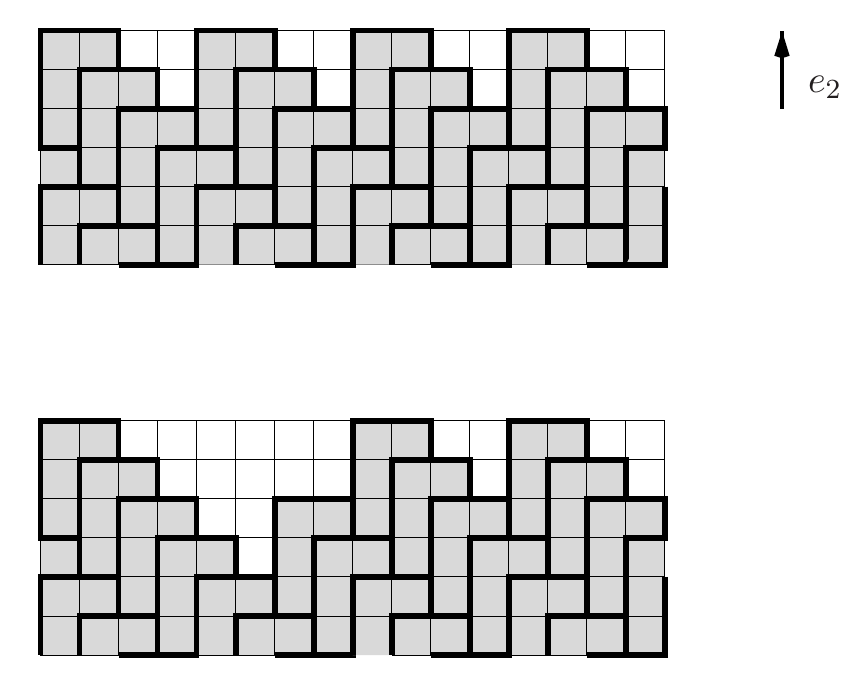}
\caption{Optimal sets in vertical directions}
\label{Fig12}
\end{figure}
We now show that $\varphi$ is a  {\em crystalline energy density} (i.e., the set $\{x: \varphi(x)= 1\}$ is a convex polygon, in this case an hexagon) determined by these six directions; i.e.,
it is linear in the cones determined by the directions. Note first that in the cones bounded by the directions ${1\over\sqrt10}(3,-1)$
and ${1\over\sqrt2}(1,-1)$, and the directions ${1\over\sqrt10}(-3,1)$ and ${1\over\sqrt2}(-1,1)$, $\varphi(\nu)=\|\nu\|_1$ since recovery sequences can be 
obtained by mixing those in Fig.~\ref{Fig11} and Fig.~\ref{Fig21}.

We then note that for $\nu=e_2$ the optimal value is a linear combination of those in $\nu={1\over\sqrt2}(\pm1,1)$,
and is obtained again by $\{E^{i,0,\nu}_h\}$ (see Fig.~\ref{Fig12}). By the convexity of $\varphi$ this implies that
$\varphi$ is linear in the cone with extreme directions ${1\over\sqrt2}(\pm1,1)$. Note that, while for ${1\over\sqrt2}(\pm1,1)$ the geometry of the interface is essentially unique, for $\nu=e_1$ this is not the case, and we may have 
non-periodic  and arbitrary oscillations of the interface (see the lower picture in Fig.~\ref{Fig12})
The symmetric argument holds for $\nu=-e_2$.

\begin{figure}[ht]
\centering
\includegraphics[width=.5\textwidth]{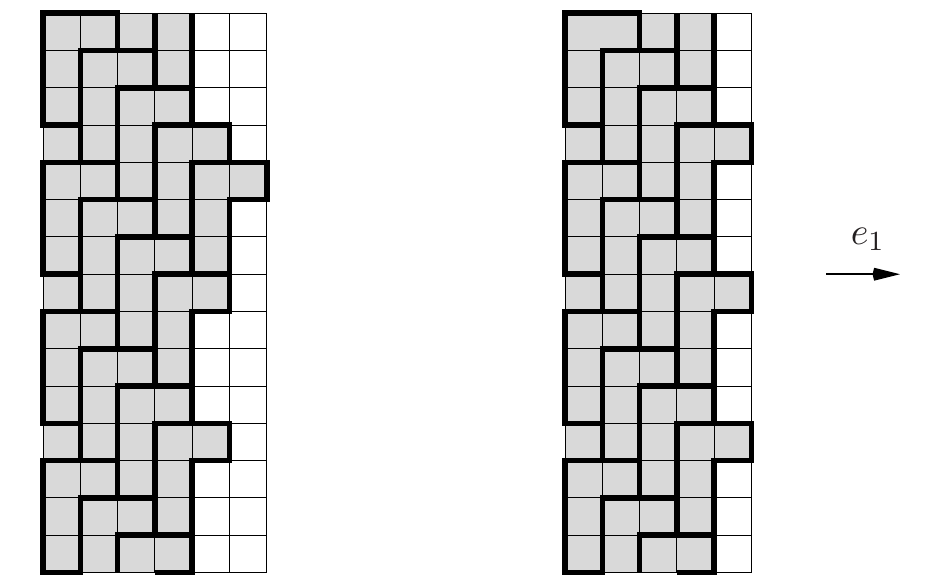}
\caption{Optimal sets in horizontal directions}
\label{Fig13}
\end{figure}
For $\nu= e_1$ the optimal value is a linear combination of those in the directions ${1\over\sqrt2}(1,1)$
and ${1\over\sqrt10}(-3,1)$, which implies that
$\varphi$ is linear in the cone with those extreme directions. Optimal sets are described in Fig.~\ref{Fig13}.
A symmetric argument gives 
the same conclusion for the opposite cone. 

Summarizing, $\varphi$ is a crystalline energy density determined by the values (using the one-homogeneous extension to $\R^2$)
$$
\varphi(1,\pm1)=\varphi(-1,\pm1)=\varphi(\pm (3,-1))= 2,
$$
\begin{figure}[ht]
\centering
\includegraphics[width=.70\textwidth]{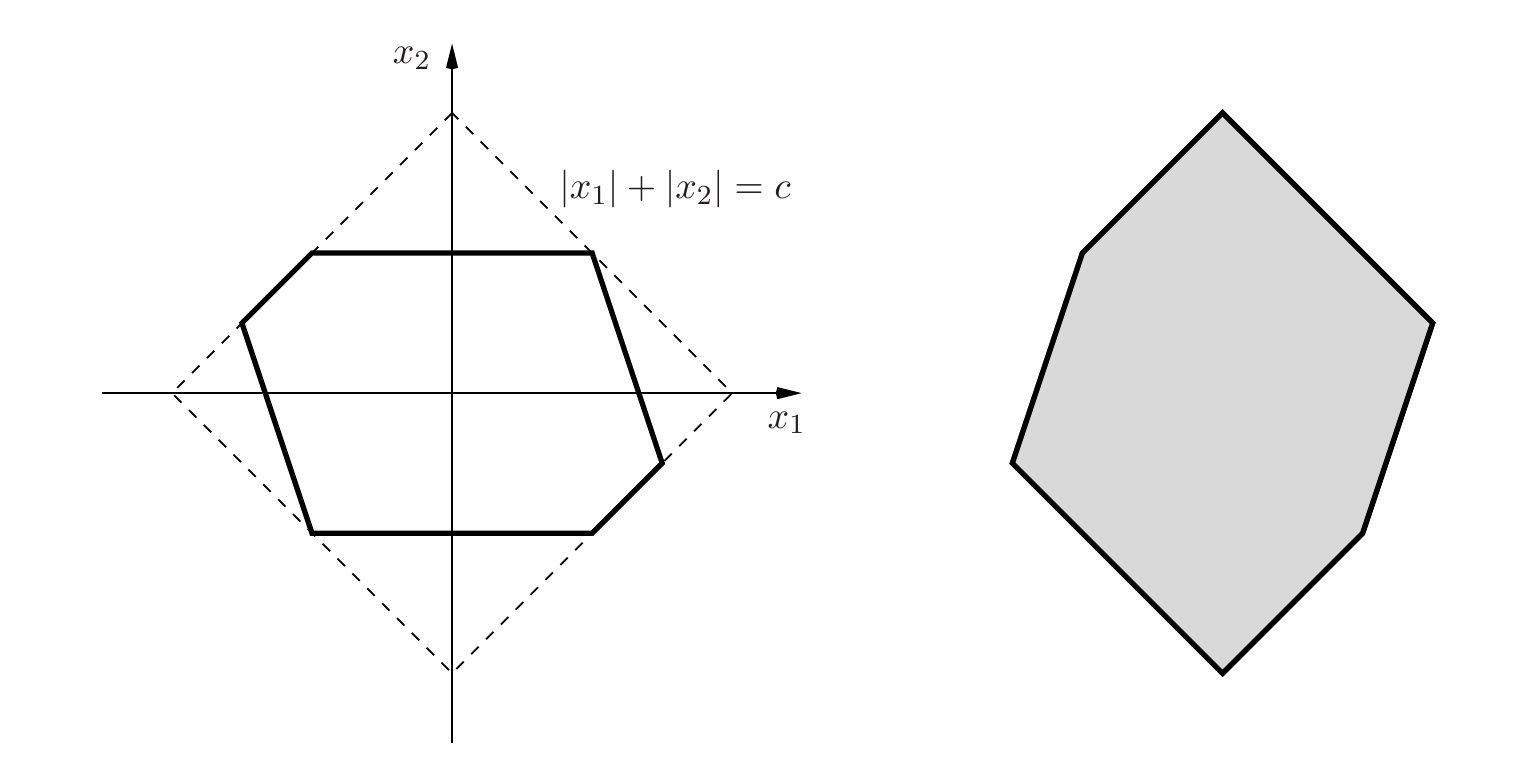}
\caption{A level set for $f(i,0,\cdot)$ and the related Wulff shape ($1\le i\le 4$)}
\label{Fig23}
\end{figure}

\begin{figure}[ht]
\centering
\includegraphics[width=.70\textwidth]{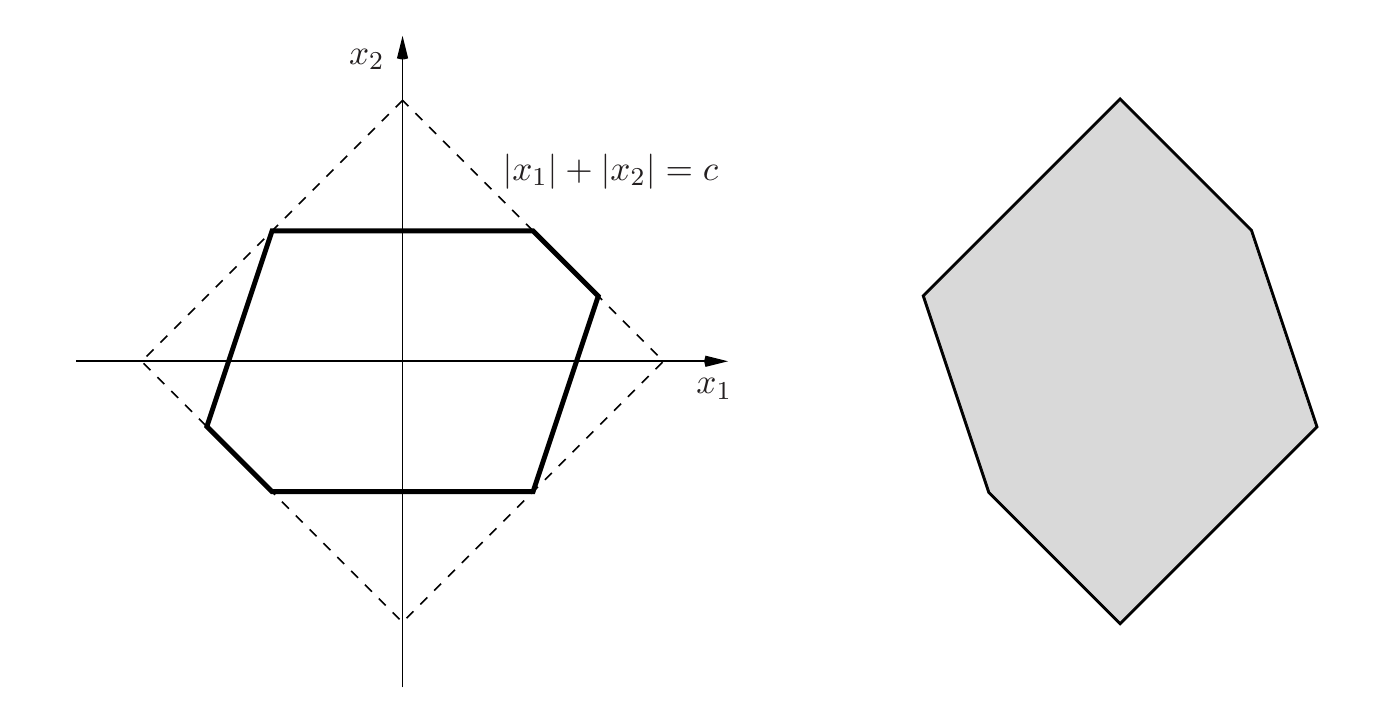}
\caption{A level set for $f(i,0,\cdot)$ and the related Wulff shape ($5\le i\le 8$)}
\label{wulff2}
\end{figure}

A level set $\{x: \varphi(x)=c\}$ is pictured on the left-hand side of  Fig.~\ref{Fig23}.
The {\em Wulff shape} related to $\varphi$ is an irregular hexagon, pictured on the right-hand side of Fig.~\ref{Fig23}. In Fig.~\ref{wulff2} we picture the corresponding sets in the case of 
$f(i,0,\cdot)$ for $5\le i\le 8$.

\smallskip\noindent
{\em Estimates}.\ \ 
From the symmetry of $f(i,0,,\cdot)$ and the subadditivity of $f$ we trivially have, for $i,j>0$ and $ i\neq j$, 
\begin{eqnarray*}
f(i,j,\nu)\le f(i,0,\nu)+ f(0,j,\nu) =\begin{cases}
2\varphi(\nu) &  \hbox{ if } i,j\le 4\cr
2\varphi(-\nu_1,\nu_2) &  \hbox{ if } i,j\ge 4\cr
\varphi(\nu)+\varphi(-\nu_1,\nu_2) &\hbox{ otherwise. }\end{cases}
\end{eqnarray*}
\begin{figure}[ht]
\centering
\includegraphics[width=.25\textwidth]{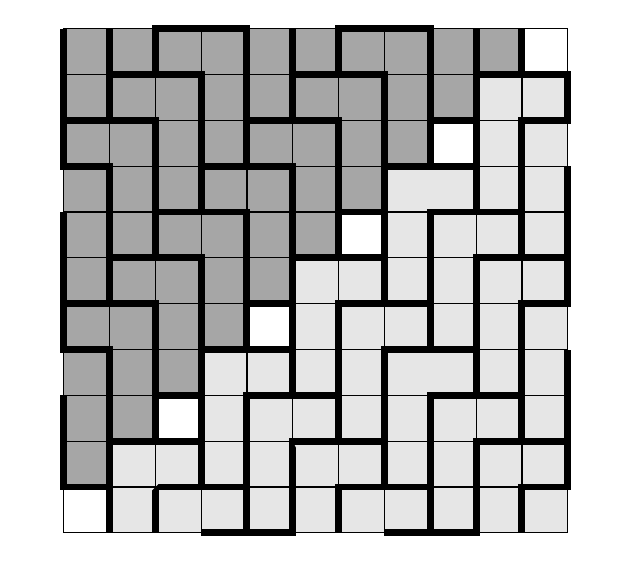}
\caption{Test sets for $i,j\neq0$}
\label{Fig14}
\end{figure}
Note however that this may be an overestimation of $f(i,j,\nu)$: in Fig.~\ref{Fig14} we exhibit test families that show that for $i=1$ and $j=7$ we have 
$$f\Bigl(i,j,{1\over\sqrt2}(1,-1)\Bigr)\le \sqrt 2= \varphi\Bigl({1\over\sqrt2}(1,-1)\Bigr).$$

\subsection{Boundary conditions}\label{bo-da}\rm
We can include in our analysis {\em anchoring} boundary conditions; i.e., we can prescribe the
trace of the elements of the partition on $\partial\Omega$.

We consider a partition $\overline A_0,\ldots,\overline  A_8$ of $\R^2$ into sets of locally finite perimeter
and suppose that for some $\eta>0$
$$
\{x\in\R^2: \hbox{dist}(x,\partial\Omega)<\eta\}\cap\bigcup_i\partial A_i
$$
is composed of a finite number of curves meeting $\partial \Omega$ transversally.
We suppose that the families $\{\overline E^{\e}_j\}$ defined by
$$
\bigcup_{i=1}^8 \{\e E_j: E_j\in\ZZ_i,\  \e E_j\subset\overline  A_i\}
$$
have equibounded energy and converge to the partition $(\overline A_0,\ldots,\overline A_8)$ on bounded sets of $\R^2$. 

The family $\{\overline E^{\e}_j\}$ can be used to define boundary conditions for $\F^\e$,
by setting 
$$
\F^\e_0(\{E^{\e}_j\})=\begin{cases}
\F^\e(\{E^{\e}_j\}{;\Omega} ) &\hbox{ if $\{E^{\e}_j\}\in \E^\e_0$,}\cr
+\infty & \hbox{otherwise}
\end{cases}
$$
where $\E^\e_0$ is the collection of families in $\{E^\e_j\}\in\E^\e$ such that 
$E^\e_j\in \{\overline E^{\e}_j\}$ if $E^\e_j\cap (\R^2\setminus\Omega_{8\e})\neq\emptyset$, where $\Omega_{\eta}=\{x\in\Omega: \hbox{dist}(x,\partial\Omega)>\eta\}$.
In particular, the sets of families $\{E^\e_j\}\in \E^\e_0$ that intersect the boundary are sets in
$\{\overline E^{\e}_j\}$. 

Then the family $\F^\e_0$ $\Gamma$-converges to
$$
\F_0(A)=\begin{cases}
\F(A) 
&\hbox{if $A\in \P_0(\Omega,\overline A)$}\cr
+\infty & \hbox{otherwise,}
\end{cases}
$$
where $\overline A=(\overline A_0,\ldots, \overline A_8)$ and 
$\P_0(\Omega,\overline A)$ is the collections of $A\in\P(\R^2)$ such that
$A_i\setminus\Omega=\overline A_i\setminus\Omega$ for all $i$.

The proof of lower bound is immediate. As the construction of recovery sequences is
concerned, given a recovery sequence for $\F(A)$ with $A\in \P_0(\Omega,\overline A)$, we can modify 
its sets close to $\partial\Omega$ as in the proof of the lower bound in Theorem \ref{thm:gammalim}.
Indeed, the proof therein deals with the case when $\Omega$ is a coordinate square 
centered in $0$ and $\overline A$ is a partition with $\{x:  \nu\cdot x=0\}$ as the unique interface.

As a consequence of the fundamental theorem of $\Gamma$-convergence \cite{GCB}
we then have that minimum values and minimizers for $\F^\e_0$ converge to the minimum value and a minimizer of $\F_0$.

\begin{remark}\label{co-min}\rm
We note that by its BV-ellipticity properties \cite{AB}, the function $f$ satisfies 
$$
{1\over |\nu_1| \vee |\nu_2|}f(l,k,\nu)=\min\Bigl\{\sum_{i<j} \int_{\overline Q\cap \partial A_i\cap \partial A_j}f(i,j,\nu^i)d\H^1: A\in \P(l,k,\nu)\Bigr\}
$$
where $\P(l,k,\nu)$ is the set of partitions $A$ such that $A=\overline A$ on $\R^2\setminus Q$ and $\overline A$ is any fixed partition such that $\overline A_l=\{x: \nu\cdot x>0\}$
and $\overline A_k=\{x: \nu\cdot x<0\}$ in an external neighbourhood of $Q$.
By the previous remark this minimum can be see as the limit as $\e\to 0$ of the minima for the corresponding approximating sequence, which can be expressed in terms of the minima $a_{ij}({1\over\e},\nu)$. By renaming $T={1\over\e}$ we obtain the limit formula for $f$
\be\label{energy-density-4}
{1\over |\nu_1| \vee |\nu_2|} f (i,j,\nu)= \lim_{T\to+\infty} 
{1\over T} 
 \min\Big\{
\F^1(\{E_k\},Q_{T})  \colon E_k\in\{E_h^{i,j,\nu}\} \hbox{ if } E_k\cap \delta Q_T\neq\emptyset
\Big\},
\ee
which proves that the $\liminf$ in \eqref{energy-density} is actually a limit.
Similarly, we obtain the limit formula \eqref{energy-density-2} 
repeating the same argument with $Q^\nu$ in the place of $Q$.
\end{remark}

\subsection{Alternate descriptions}
From Theorem \ref{thm:gammalim} we can derive descriptions for the limit of the energies
$\F^\e$ with respect to other types of convergence.

We can consider the energies $\F^\e$ as defined on sets $\A^\e(\Omega)$. We define
\begin{equation}\label{spin-form}
\F^\e_{\Omega,s}(E)= \H^1(\Omega\cap\partial E) \hbox{ if } E=\bigcup_j E_j \hbox{ and }  \{E_j\}\in \E^\e(\Omega).
\end{equation}
\vskip-.1cm\noindent
The subscript $s$ stands for ``spin''. By this notation we imply that we
regard a union of molecules as a constrained spin system and we
do not wish to distinguish between different types of molecules. 
We then may consider the convergence $E^\e\to E$, defined as $|E^\e\triangle E|\to 0$
as $\e\to 0$, for which the sequence $\F^\e$ is equi-coercive. 
From Theorem \ref{thm:gammalim} we deduce the following result.

\begin{theorem}\label{thisp}
Let $\F_{\Omega,s}^\e$ be defined by \eqref{spin-form}. Then the $\Gamma$-limit of $\F_{\Omega,s}^\e$ 
with respect to the convergence $E^\e\to E$ is defined on sets of finite perimeter $E$ 
by 
\begin{equation}\label{spin-form-lim}
\F_{\Omega,s}(E)= \inf\Biggl\{ \F_\Omega(A): A= (A_0,\ldots, A_8)\in \P(\Omega), \bigcup_{i=1}^8 A_i=E\Biggr\}.
\end{equation}
\end{theorem}

\begin{proof} In order to prove the lower bound it suffices to remark that if 
$\sup_\e\F^\e_{\Omega,s}(E_\e)<+\infty$ and $E_\e\to E$ then, up to subsequences,
we can decompose $E_\e=\bigcup_{i=1}^8 E^i_\e$ with $E^i_\e\to A^i$, so that
$$
\liminf_{\e\to 0} \F^\e_{\Omega,s}(E_\e)\ge \F_\Omega(A)\ge \F_{\Omega,s}(E).
$$
In order to prove the upper bound, we may suppose that the infimum in 
\eqref{spin-form-lim} is achieved by some $A= (A_0,\ldots, A_8)$ with 
$\bigcup_{i=1}^8 A_i=E$. We then take a recovery sequence $\{E^\e_j\}$ for
$\F_\Omega(A)$, and a recovery sequence for $\F_{\Omega,s}(E)$ is 
then given by $E^\e=\bigcup_jE^\e_j$.
\end{proof}

\begin{remark}[non-locality of the $\Gamma$-limit]\rm
The $\Gamma$-limit $\F_{\Omega,s}(E)$ 
cannot be represented as an integral on $\partial E$. 

To check the non locality, consider as an example the target set $E$ obtained as
the intersection of the two Wulff shapes in Fig.~\ref{Fig23} and Fig.~\ref{wulff2}. If it were
local then the optimal microstructure close to an edge with normal ${1\over\sqrt{10}}(3,-1)$
should be composed of molecules in some $\ZZ_i$ with $i\in\{1,\ldots,4\}$, while 
the optimal microstructure close to an edge with normal ${1\over\sqrt{10}}(3,1)$
should be composed of molecules in some $\ZZ_i$ with $i\in\{5,\ldots,8\}$. 
This implies that the optimal $A_1,\ldots, A_8$ must have at least two non-empty
sets, and the value of $\F_{\Omega,s}(E)$ depends on an interface not localized
on $\partial E$.

Note more in general that we can give a local lower bound by optimizing the surface energy density
at each fixed value of $\nu$.
Namely, if we define
$$
f(x)=\min\{f(i,0,x): i\in\{1,\ldots,8\}\},
$$
then a lower bound for $\F_{\Omega,s}(E)$ is given by
$$
\F_{\Omega,s}(E)\ge \int_{\overline\Omega\cap\partial E} f^{**}(\nu)d\H^1,
$$
where $\nu$ is the inner normal to $E$ and $f^{**}$ is the convex envelope of $f$ \cite{AB}.
Note that this estimate derives from \eqref{spin-form-lim} by neglecting
interfacial energies in $\F_\Omega(A)$ which are internal to $E$; i.e., those 
corresponding to $\partial A_i\cap\partial A_j$ with $i,j> 0$.
By the computations of Section \ref{descrip} we can give an explicit description of $f^{**}$,
since it is positively one homogeneous and its level set $\{x: f^{**}(x)=1\}$ is the convex envelope
of the union of the two corresponding level sets for $f(1,0,\cdot)$ and $f(5,0,\cdot)$ in 
Fig.~\ref{Fig23} and Fig.~\ref{wulff2}. Since for some $\nu$ (e.g., $\nu=(1,0)$) 
we have $f^{**}(\nu)<f(\nu)$,
for such $\nu$ the optimal interface would be obtained by a surface microstructure
with both $S$ and $R$ molecules, which is not possible without introducing additional 
surface energy corresponding to some $\partial A_i\cap\partial A_j$
with $i\in\{1,\ldots,4\}$ and $j\in\{5,\ldots,8\}$. This shows that the lower bound is not sharp for 
example for sets with a vertical part of the boundary. 
\end{remark}

Another possibility is to consider the two types of molecules $R$ and $S$  as parameters; i.e.,
rewrite the energy as 
\begin{eqnarray}\label{RS-form}
&&\F^\e_{\Omega,R,S}(E_R, E_S)= \H^1(\Omega\cap\partial (E_R\cup E_S) ) \\
&&\hbox{ if } 
E_R=\bigcup \Bigl\{E_j: E_j\in\bigcup_{i=1}^4\ZZ^\e_i\Bigr\},\quad E_S=\bigcup  \Bigl\{E_j: E_j\in\bigcup_{i=5}^8\ZZ^\e_i\Bigr\}\hbox{ and }  \{E_j\}\in \E^\e(\Omega),\nonumber
\end{eqnarray}
and consider the convergence $(E^\e_R,E^\e_S)\to (E_R, E_S)$ defined as the separate convergence $E^\e_R\to E_R$ and  $E^\e_S\to E_S$. Note that also this convergence is 
compact by Theorem \ref{te-co}, since $E^\e_R=\Omega^\e_1\cup\cdots\cup \Omega^\e_4$
and $E^\e_S=\Omega^\e_5\cup\cdots\cup \Omega^\e_8$ in the notation of Definition \ref{def-convergenza}. We then have the following result, whose proof is essentially the same as
that of Theorem \ref{thisp}.

\begin{theorem}\label{thRS}
Let $\F_{\Omega,R,S}^\e$ be defined by \eqref{spin-form}. Then the $\Gamma$-limit of $\F_{\Omega, R,S}^\e$ 
with respect to the convergence $(E^\e_R,E^\e_S)\to (E_R,E_S)$ is defined on 
pairs of sets of finite perimeter $(E_R,E_S)$ 
by 
\begin{equation}\label{RS-form-lim}
\F_{\Omega,R,S}(E_R,E_S)= \inf\Biggl\{ \F_\Omega(A): A= (A_0,\ldots, A_8), \bigcup_{i=1}^4 A_i=E_R,
\bigcup_{i=5}^8 A_i=E_S\Biggr\}.
\end{equation}
\end{theorem}

\begin{remark} \rm
We can give a lower bound of $\F_{\Omega,R,S}$ by an interfacial energy 
by interpreting this functional as defined on partitions of $\Omega$ into three sets of finite 
perimeter $(A_0,A_1,A_2)$ where $A_1=E_R$, $A_2=E_S$ and 
$A_0=\Omega\setminus (E_R\cup E_S)$. By Theorem \ref{thm:gammalim}
and a minimization argument we have
\begin{eqnarray}\label{gamma-lim-RS}
\F_{\Omega,R,S}(E_R,E_S)&\ge&
\int_{\Omega\cap\partial E_R\cap\partial A_S} f_0^{**}(\nu^R)d\H^{1}\\
&&+
\int_{\overline\Omega\cap(\partial E_R\setminus\partial E_S)} f_R(\nu^R)d\H^{1}+
\int_{\overline\Omega\cap(\partial E_S\setminus\partial E_R)} f_S(\nu^S)d\H^{1},\nonumber
\end{eqnarray}
where $f_R(\nu)=f(1,0,\nu)$, $f_S(\nu)= f(5,0,\nu)$ and
$$
f_0(\nu)=\min\{f(i,j,\nu): i\in\{1,\ldots,4\}, j\in\{5,\ldots,8\}\},
$$
and $\nu^R$ and $\nu^S$ are the inner normals to $E_R$ and $E_S$, respectively.
Note however that the right-hand side in \eqref{gamma-lim-RS}
may not be a lower-semicontinuous functional
on partitions, and hence should be relaxed taking a BV-elliptic envelope \cite{AB}.
This computation would be necessary to check if this lower bound is actually sharp
so that the functional $\F_{\Omega,R,S}$ is local. Unfortunately the computation of a 
BV-elliptic envelope is in general an open problem, and cannot be reduced to a 
computation of a convex envelope as in the case of a single set of finite perimeter.
\end{remark}

\section{Generalizations and remarks}\label{Sect4}
1. We can consider an inhomogeneous dependence for the surface energy.
As an example, we can fix two positive constants $c_R$ and $c_S$ and 
consider the functionals
$$
\F^\e(\{E_j\};\Omega)=c_R\H^1(\Omega\cap \partial_R E)+
c_S\H^1(\Omega\cap \partial_S E),
$$
where $E=\bigcup_j E_j$ and
$$
\partial_R E=\partial E\cap \partial\bigcup\Bigl\{E_j: E_j\in \ZZ_1\cup\cdots\cup\ZZ_4\Bigr\}
$$
$$
\partial_S E=\partial E\cap \partial\bigcup\Bigl\{E_j: E_j\in \ZZ_5\cup\cdots\cup\ZZ_8\Bigr\}
$$
The result is the same, upon defining the surface densities using the corresponding $\F^1$.

Note that is this case, when computing $f(i,0,\nu)$, we might have a {\em wetting} phenomenon; i.e., the presence of
a layer of a different phase at the boundary of another. This is clear if for example 
$c_S$ is sufficiently smaller than $c_R$ at the boundary between the phase $0$ and the phase~$1$.
\begin{figure}[ht]
\centering
\includegraphics[width=.25\textwidth]{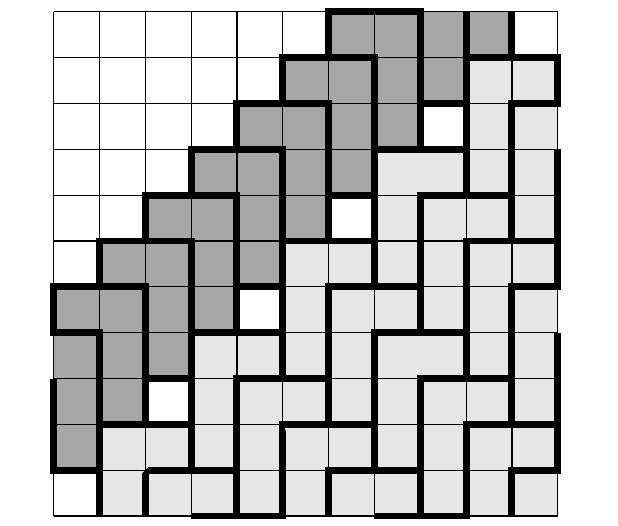}
\caption{A wetting layer}
\label{Fig15}
\end{figure} 
In Fig.~\ref{Fig15} we picture a configuration giving the estimate 
$$
f\Bigl(1,0,{1\over\sqrt2}(-1,1)\Bigr)
\le {3\over 2} c_S \sqrt 2 +{1\over 2} c_R\sqrt 2,
$$
which is energetically convenient with respect to the one in Fig.~\ref{Fig11} if $3c_S<c_R$.

\smallskip
2. In the whole analysis we can replace the surface energy by a scaled volume energy
\begin{equation}\label{volumen}
\F^\e(\{E_j\};\Omega)={1\over\e} |\Omega\setminus E|.
\end{equation}
Note that in this case the empty phase disappears by the definition 
of the energy, so that the $\Gamma$-limit is finite only if $A_0=\emptyset$.

The simplest case is $\Omega=\R^2$. In this case the proof proceeds exactly in the 
same way. Indeed, the argument of Lemma \ref{quadrato-interno} is independent of
energy arguments, while in the compactness Theorem \ref{te-co} the equi-boundedness
of the energy is used to obtain estimate \eqref{stir}, which follows in an even easier way
under the assumption of the equiboundedness of the energies \eqref{volumen}.

The surface densities are then defined by 
\be\label{energy-density-5}
f (i,j,\nu)= (|\nu_1| \vee |\nu_2|) \lim_{T\to+\infty} 
{1\over T} 
 \min\Big\{\Bigl|Q_T\setminus \bigcup\{E_k\}\Bigr|
  \colon E_k\in\{E_h^{i,j,\nu}\} \hbox{ if } E_k\cap \delta Q_T\neq\emptyset
\Big\},
\ee
since we simply have $\F^1(\{E_k\},\Omega)=|\Omega\setminus \bigcup\{E_k\}|$.

The case $\Omega\neq\R^2$ cannot be treated straightforwardly as above,
since the approximation argument of $\Omega$ by polyhedral sets in the proof
of the upper inequality cannot be used. We conjecture that a different boundary
term on $\partial\Omega$ arises, taking into account approximations
of $\Omega$ with minimal two-dimensional measure. Note that even when
the approximation is not constrained to be performed with unions of molecules
this may be a delicate numerical issue~\cite{Ros}.

\smallskip
3. We can give a higher-order description of our system by scaling the energy as
$$
\F^\e_1(\{E_j\};\Omega)={1\over\e}\H^1(\Omega\cap \partial E). 
$$
In this case the limit is finite only at minimizers of $\F(A;\Omega)$.

(a) If $\Omega\neq \R^2$ then the only minimizer is given by $A_0=\Omega$
and $A_i=\emptyset$ for $i>0$. Sequences with equi-bounded energy are 
$\{E^\e_j\}$ with $\sup_\e\sharp\{E^\e_j\}<+\infty$. We can then define the 
convergence $\{E^\e_j\}\to ((x_1,r_1,s_1),\ldots, (x_N,r_N,s_N))$, where
$x_k$ are the limit points of sequences in $\{E^\e_j\}$, $r_k$ is the number of
molecules of the type $\e R(n)$ in $\{E^\e_j\}$ converging to $x_k$ and $s_k$ is the number of
molecules of the type $\e S(n)$ in $\{E^\e_j\}$ converging to $x_k$.
The $\Gamma$-limit is then defined by
$$
\F_1(\{(x_k,r_k,s_k)\}_k)=\sum_k \phi(r_k, s_k),
$$
\vskip-.3cm
\noindent where 
\begin{eqnarray*}
\phi(r,s)&=&\min\Bigl\{ \H^1(\partial E): E=\bigcup E_j, \{E_j\}_j \hbox{ disjoint family }\\
&&\qquad\hbox{ composed of } r \hbox{ sets } R(n_k)
\hbox{ and } s \hbox{ sets }  S(m_l)\Bigr\}
\end{eqnarray*}

(b) If $\Omega= \R^2$ then we have the nine minimizers with $A_i=\R^2$
for some $i$ and $A_k=\emptyset$ for $k\neq i$. 
For $i=0$ we are in the same case as above. Otherwise, we can suppose that
$i=1$. We can consider the convergence $\{E^\e_j\}\to ((x_1,s_1),\ldots, (x_N,s_N))$
with $s_k$ defined as above. The $\Gamma$-limit is then defined by
$$
\F_1(\{(x_k,s_k)\})=\sum_k \phi(s_k),
$$
\vskip-.3cm
\noindent where 
\begin{eqnarray*}
\phi(s)&=&\min\Bigl\{ \H^1(\partial E): E=\bigcup E_j, \{E_j\}_j \hbox{ disjoint family composed of}\\
&&\qquad s \hbox{ sets }  S(m_l)\hbox{ and of all the }E_j\in \ZZ_1\hbox{ outside a compact set} \Bigr\}.
\end{eqnarray*}
We can conjecture that the minimizers of this problem are given by an array
of $s$ sets $E_j$ in the same $\ZZ_i$ for some $i\ge 4$ surrounded by elements in $\ZZ_1$.

\smallskip
4. The analysis of the functionals $\F^\e$ is meaningful also if only one type of molecule
is taken into account. In this case we have only four modulated phases and the limit is defined on partitions into sets of finite perimeter indexed by five parameters. The proof
follows in the same way, with the interfacial energies defined by using families 
composed only of the type of molecule considered.

It is interesting to note that this remark applies also if we take into account only one type
of molecules in the pair on right-hand side of Fig.~\ref{more-mol}. Indeed, in that case
there is a single pattern for the ground states with four modulated phases and 
Lemma \ref{quadrato-interno} holds (while we have already remarked that it does not
hold if we consider ensembles of both molecules in that pair). 
On the contrary, for a single type of molecules in the pair on 
left-hand side of Fig.~\ref{more-mol}, it is possible to construct infinitely many different
structures with zero energy composed of stripes of the same two-periodic structure
(see the ones in the first picture in Fig.~\ref{no-lemma})
with arbitrary vertical shifts. Hence, it is not possible to reduce to a single pattern (or a finite
number of patterns) for the ground states.

\goodbreak
\bigskip\noindent{\bf Acknowledgments.}
The subject of this paper was inspired by a MoMA Seminar by G.~Contini at Sapienza University in Rome. Part of this work was elaborated in 2014 while the first two authors were visiting the Mathematical Institute in Oxford, whose kind hospitality is gratefully acknowledged.

\goodbreak

\end{document}